\documentclass[article]{biometrika}

\usepackage{amsmath,amssymb,latexsym}
\usepackage[utf8]{inputenc}

\usepackage{times}
\usepackage{bm}
\usepackage{natbib}
\usepackage[plain,noend]{algorithm2e}
\usepackage{comment}
\usepackage{url}

\makeatletter
\renewcommand{\algocf@captiontext}[2]{#1\algocf@typo. \AlCapFnt{}#2} 
\def\@algocf@capt@plain{top}
\renewcommand{\algocf@makecaption}[2]{%
  \addtolength{\hsize}{\algomargin}%
  \sbox\@tempboxa{\algocf@captiontext{#1}{#2}}%
  \ifdim\wd\@tempboxa >\hsize
    \hskip .5\algomargin%
    \parbox[t]{\hsize}{\algocf@captiontext{#1}{#2}}
  \else%
    \global\@minipagefalse%
    \hbox to\hsize{\box\@tempboxa}
  \fi%
  \addtolength{\hsize}{-\algomargin}%
}
\makeatother

\usepackage{xspace}

\newcommand{\ABCG}{ABC-Gibbs\xspace}
\newcommand{\ABCS}{ABC\xspace}
\newcommand{\ABC}{Approximate Bayesian computation\xspace}

\renewcommand{\epsilon}{\varepsilon}
\newcommand{\btheta}{{\theta}}
\newcommand{\bmu}{\mu}
\newcommand{\Ntot}{N_{\text{tot}}}
\newcommand{\U}{\mathcal{U}}

\usepackage{tikz}
\usetikzlibrary{arrows}

\graphicspath{{../}}

\usepackage{listings}

\begin{document}




\title{Component-wise Approximate Bayesian Computation via Gibbs-like steps}
\author{Gr\'egoire Clart\'e, Christian P.~Robert, Robin J. Ryder, and Julien Stoehr}

\affil{Universit\'e Paris-Dauphine, Universit\'e PSL, CNRS, CEREMADE, 75016 Paris, France}

\maketitle

\begin{abstract}
Approximate Bayesian computation methods are useful for generative models with intractable likelihoods. These methods are however sensitive to the dimension of the parameter space, requiring exponentially increasing resources as this dimension grows. To tackle this difficulty, we explore a Gibbs version of the \ABC approach that runs component-wise approximate Bayesian computation steps aimed at the corresponding conditional posterior distributions, and based on summary statistics of reduced dimensions. While lacking the standard justifications for the Gibbs sampler, the resulting Markov chain is shown to converge in distribution under some partial independence conditions. The associated stationary distribution can further be shown to be close to the true posterior distribution and some hierarchical versions of the proposed mechanism enjoy a closed form limiting distribution. Experiments also demonstrate the gain in efficiency brought by the Gibbs version over the standard solution.
\end{abstract}



\section{Introduction}

Approximation Bayesian computation (ABC) is a computational method which stemmed from
population genetics to deal with intractable likelihoods, that is models whose likelihood cannot be (easily) computed
 but which can be simulated from \citep{tavare1997,
beaumont2002}. Since then, it has been  applied to numerous other fields: see for
example \cite{toni2008, csillery2010, moores2015, handbook}. The principle of
the method is to simulate pairs of parameters and pseudo-data from the prior
predictive, keeping only the parameters that bring the pseudo-data  close enough (within a pseudo-distance $\varepsilon$) to the
observed data. Proximity is often defined in terms of a projection
of the data, called a summary statistic. In general, practitioners of
ABC aim to use informative summary statistics and select $\varepsilon$ to be as
small as possible, since this leads to a higher-quality approximation.
From the start, this method has
suffered from the curse of dimensionality in that the dimension of the
parameter to be inferred imposes a lower bound on the dimension of the
corresponding summary statistic to be used (results by 
\cite{fearnhead2012} and \cite{li2018} imply that the dimension 
of the summary statistic should be identical to the dimension of the parameter). 
This constraint impacts the range of the distance between observed and 
simulated summaries, with the distance choice having a growing impact as 
the dimension increases. Reducing the dimension of the summary is thus 
impossible without reducing the dimension of the parameter, which sounds 
an impossible goal unless one infers about one parameter at a time, 
suggesting a Gibbs sampling strategy where a different and much reduced 
dimension summary statistic is used for each component of the parameter.
The purpose of this paper is to explore and validate this strategy, producing
sufficient conditions for the convergence of the resulting algorithms.

Additionally, the Gibbs perspective allows us to account for the current values of
the other components of the parameter and therefore to shy away from simulating
from the prior which is an inefficient proposal. This feature
connects this proposal with earlier solutions in the literature such as the
Metropolis version of \cite{marjoram2003} and the various sequential Monte
Carlo schemes \citep{toni2008, beaumont2009}. There have been earlier ABC versions with 
Gibbs features, including \cite{wilkinson:2011}, where a two-stage
ABC-within-Gibbs algorithm is proposed towards bypassing the intractibility of
one of the conditional distributions used in their Gibbs sampler. Since the
other conditional distribution is simulated exactly, there is no convergence issue with this
version. Note also that the summary statistics used in that paper are not
chosen for dimension reduction purposes. \cite{kousathanas:2016} also run a Gibbs-like
ABC algorithm that assumes the availability of conditionally sufficient
statistics to preserve the coherence of the algorithm. \cite{rodrigues2019likelihood} propose another 
Gibbs-like ABC algorithm in which the conditional distributions are approximated by regression models.

A Gibbs version of the ABC method offers a range of potential improvements
compared with earlier versions, induced in most cases by the dimension
reduction thus achieved. First, in hierarchical models, conditioning
decreases the number of dependent components, and some of the conditionals may
be available in closed form, which makes the approach only semi-approximate. Second,
since the conditional targets  live in spaces of low dimension, they can
more easily be parametrised by low dimension functions of the conditioning
terms. This justifies using a restricted range of collection of statistics,
which may in addition depend on other parameters. Third, reducing the dimension
of the summary statistic improves the approximation since a
smaller tolerance can then be handled at a manageable computing cost.

This heuristic leads us to propose in Section \ref{Intro} a generic algorithm called \ABCG. To show the theoretical validity of this idea, we successively show that, under some conditions:
\begin{enumerate}
\item[i)] for all $\varepsilon>0$, our \ABCG converges to a certain limiting distribution $\nu_\varepsilon$ in total variation distance,
\item[ii)] when $\varepsilon \rightarrow 0$, $ \Vert \nu_\varepsilon - \nu_0 \Vert_{TV} \rightarrow 0$, with $\nu_0$ a distribution,
\item[iii)] $\nu_0$ is the limiting distribution of Vanilla \ABCS with tolerance $\varepsilon$ set to $0$.
\end{enumerate}


The first result corresponds to Theorem \ref{THgen} in the general case; Theorem \ref{thm:k1k2k3} states this result for hierarchical models under looser assumptions. The second result is a consequence of Theorem \ref{THgenconv}. The last result follows from the results of Section \ref{sec:comp}.

\section{Approximate Bayesian Gibbs sampling}
\label{Intro}

%
%
%

\subsection{Vanilla approximate Bayesian computation}
\label{ABCSim}

Approximate Bayesian computation methods, summarised in Algorithm \ref{ABCSimAlg}, provide a technique to sample posterior distributions when the corresponding likelihood $f(x|\btheta)$ is intractable, that is the numerical value $f(x|\btheta)$ cannot be computed in a reasonable amount of time, but the model is generative, that is it allows for the generation of synthetic data given a value of the parameter. Given a prior distribution on the parameter $\btheta$, it builds upon samples from the associated prior predictive $(\btheta^{(i)}, x^{(i)})_{i=1,\ldots,N}$ by selecting pairs such that the pseudo-data $x^{(i)}$ stand in a neighbourhood of the observed data $x^\star$. 

Since both the simulated and observed dataset may belong to a space of a high dimension, the neighbourhood is usually defined with respect to a summary statistic $s(\cdot)$ of a lesser dimension and an associated distance $d$ (see \citealp{Marin2012} for a review). \citet{fearnhead2012} show that the optimal statistic is of the same dimension as the parameter $\btheta$; in practice, the choice of $s$ remains a crucial issue.


\begin{algorithm}[!h]
\SetAlgoLined
\KwIn{observed dataset $x^\star$, number of iterations $N$, threshold $\varepsilon>0$, summary statistic $s$.}
\KwOut{a sample $(\btheta^{(1)}, \ldots, \btheta^{(N)})$.}
\For{$i = 1 , \dots, N$}{
 \textbf{repeat}

 $\btheta^{(i)} \sim \pi(\cdot)$
 
 $x^{(i)} \sim f(\cdot \mid \btheta^{(i)})$
 
 \textbf{until} $d\{s(x^{(i)}),s(x^\star)\}<\varepsilon$
}
\caption{Vanilla \ABC}
\label{ABCSimAlg}
\end{algorithm}

The output of Algorithm \ref{ABCSimAlg} 
is a sample distributed from an approximation of the posterior \citep{tavare1997,handbook}. 
Its density is written, with a notation coherent with the next sections: 
\[
\pi_{\varepsilon}\{\btheta \mid s(x^\star)\} \propto \int \pi(\btheta) f(x \mid \btheta) 
\mathbf{1}_{d\{s(x),s(x^\star)\}<\varepsilon}\,\mathrm{d}x.
\]
This approximation depends on the choice of both the summary statistic  $s$ and the tolerance
level $\varepsilon$. \cite{frazier2018} show its consistency, namely that when the number of
observations tends to $\infty$ and the tolerance tends to $0$ at a proper relative rate, the
approximate posterior concentrates at the true value of the parameter,
albeit as a posterior distribution associated with the statistic $s$, rather
than the true posterior, when $s$ is not sufficient. The shape of the asymptotic distribution is further discussed in \cite{li2018} and \cite{frazier2018}. 

More to the point, given a fixed number of observations, the approximate posterior also
converges to the  posterior $\pi\{\btheta \mid s(x^\star)\}$, rather than
to the standard posterior $\pi(\btheta \mid x^\star)$, when the tolerance level
goes to $0$. In practice, however, the tolerance level cannot be equal to zero and is
customarily chosen as a simulated distance quantile \citep{handbook}. In practice,
a large sample of pseudo-observations is generated from the prior predictive and the corresponding distances to the observations are computed. We use the term reference table for this collection of parameters and distances. The tolerance
is then derived as a small quantile of these distances.


\subsection{Gibbs sampler}
\label{GibbsSpres}

The Gibbs sampler, first introduced by \cite{geman:1984} and generalised by
\cite{GelflandSmith}, is an essential element in Markov chain Monte Carlo
methods \citep{robert:casella:2004,gelman2013bayesian}. As described in Algorithm
\ref{Gibbssimple}, for a parameter $\btheta=(\theta_{1},\dots,\theta_{n})$,
it produces a Markov chain associated with a given target joint distribution,
denoted $\pi$, by alternatively sampling from each of its conditionals.


\begin{algorithm}[!h]
\KwIn{number of iterations $N$, starting point $\btheta^{(0)}= (\theta^{(0)}_1,\dots,\theta^{(0)}_n)$.}
\KwOut{a sample $(\btheta^{(1)},\dots,\btheta^{(N)})$.}
%

\For{$i=1,\ldots,N$}{
  \For{$j=1,\ldots,n$}{
    $\theta_j^{(i)}\sim\pi(\cdot\mid \theta_1^{(i)},\ldots,\theta_{j-1}^{(i)},\theta_{j+1}^{(i-1)},\ldots,\theta_{n}^{(i-1)})$
  }
}
\caption{Gibbs sampler }
\label{Gibbssimple}
\end{algorithm}

Gibbs sampling is well suited to high-dimensional situations where the
conditional distributions are easy to sample. In particular, as illustrated by
the long-lasting success of the BUGS software \citep{lunn:bugs:2012}, hierarchical Bayes
models often allow for simplified conditional distributions thanks to partial
independence properties. Considering for instance the common hierarchical model \citep{lindley:smith:1972,carlin:louis:1996} defined by
\begin{equation}\label{eq:modelun}
x_{j} \mid \mu_{j} \sim \pi(x_{j} \mid \mu_{j})\,,\qquad
\mu_{j} \mid \alpha \stackrel{\text{i.i.d.}}{\sim} \pi(\mu_{j} \mid \alpha)\,,\qquad
\alpha \sim \pi(\alpha)\,.
\end{equation}
The joint posterior of  $\bmu = (\mu_{1},\dots ,\mu_{n})$ conditional on $\alpha$ then factorises as 
\[
\pi(\bmu \mid x_{1},\dots ,x_{n}, \alpha) \propto
\prod_{j=1}^n\pi(\mu_{j} \mid \alpha)\pi(x_{j} \mid \mu_{j}).
\]
This implies that the full conditional posterior of a given $\mu_{j}$ only depends on $\alpha$
and $x_{j}$, independently of the other $(\mu_{\ell}, x_{\ell})$'s.



\subsection{Component-wise Approximate Bayesian Computation}
\label{AGibbs}

When handling a model such as \eqref{eq:modelun} with both a high-dimensional parameter and an intractable
likelihood, the Gibbs sampler cannot be implemented, while the vanilla
ABC sampler is highly inefficient. This curse of dimensionality attached to the ABC algorithm is well documented
\citep{li2018}.

Bringing both approaches together may subdue this loss efficiency, by
sequentially sampling from the ABC version of the conditionals, whose density 

$\pi_{\varepsilon_{j}}(\cdot \mid s_{j}(x^\star, \theta^{(i)}_1,\dots,\theta^{(i)}_{j-1},\theta^{(i-1)}_{j+1},\dots,\theta^{(i-1)}_{n}))$ 
is proportional (see \eqref{ABCSim}) to 
\[ \int \pi( \theta_j \mid \theta^{(i)}_1,\dots,\theta^{(i)}_{j-1},\theta^{(i-1)}_{j+1},\dots,\theta^{(i-1)}_{n}) f(x \mid \theta^{(i)}_1,\dots,\theta^{(i)}_{j-1},\theta_j, \theta^{(i-1)}_{j+1},\dots,\theta^{(i-1)}_{n}) \mathbf{1}_{d\{s_j(x),s_j(x^\star)\}<\varepsilon_j}\,\text{d}x\,.\]

Each step in
Algorithm \ref{Gibbssimple} is then replaced by a call to Algorithm
\ref{ABCSimAlg}, conditional on the other components of the parameter.  We
obtain a generic componentwise approximate Bayesian computational method,
summarised as Algorithm \ref{AGgen}.

This algorithm can be analysed as
a variation of Algorithm \ref{ABCSimAlg} in which the synthetic data $x^{(i)}$ are
simulated from the conditional posterior predictive, rather than from the prior
predictive. This may result in simulating both parameters and pseudo-data component-wise from spaces of smaller
dimension. This also allows the use of statistics of lower dimension, as exemplified in Section \ref{Nonhier}. 

Each stage $j$ of the algorithm now requires its own tolerance level $\varepsilon_j$ and statistic $s_j$. This statistic can be a function of the observations, but also of the other parameters $(\theta_i)_{i\neq j}$ which are conditioned upon at stage $j$. Typically, $\theta_j$ is of dimension 1 and so $s_j$ should also be of dimension 1, per the results of \citet{fearnhead2012}. Finding a good unidimensional statistic for each $\theta_j$ in \ABCG may prove easier than finding a good high-dimension statistic for Vanilla ABC.

\begin{algorithm}[!h]
\KwIn{number of iterations $N$, starting point $\btheta^{(0)}=(\theta^{(0)}_1,\dots, \theta^{(0)}_n)$, 
thresholds $\varepsilon=(\varepsilon_{1},\dots,\varepsilon_{n})$, statistics $s_{1},\dots, s_{n}$, distances on the statistics $d_1,\dots,d_n$.}
\KwOut{a sample $(\btheta^{(1)},\ldots,\btheta^{(N)})$.}
\For{$i = 1 , \dots, N$}{
\For{$ j= 1,\dots ,n$}{
 $\theta^{(i)}_j \sim \pi_{\varepsilon_{j}}\{\cdot \mid s_{j}(x^\star, \theta^{(i)}_1,\dots,\theta^{i}_{j-1},\theta^{(i-1)}_{j+1},\dots,\theta^{(i-1)}_{n})\}$
 }
}
\caption{ABC-Gibbs} 
\label{AGgen}
\end{algorithm}

If $\varepsilon_{j}=0$ and if $s_{j}$ is a conditionally sufficient statistic,
the corresponding $j$th step in Algorithm \ref{AGgen} is an exact simulation
from the corresponding conditional. In particular, if some of the conditional
distributions can be perfectly simulated, this cancels the need for an
approximate step in the algorithm.  In practice, to simulate from the
approximate conditional, and similarly to Algorithm \ref{ABCSimAlg}, we take
$\varepsilon_{j}$ as an empirical distance quantile.  In other words, for the
$j$th component of the parameter, conditional on the other components, we
simulate a small reference table from its conditional prior and output the
parameter associated with the smallest distance. 

At first, the purpose of this algorithm may sound unclear as the limiting
distribution and its existence are unknown. 
As shown in Theorem \ref{THgen}, convergence can indeed be achieved,
based on a simple condition.  For simplicity's sake, we initially only consider the case
when $n=2$ in Algorithm \ref{AGgen}.



\begin{theorem}
\label{THgen}
Assume that there exists $0<\kappa<1/2$ such that
\[
\sup_{\theta_{1},\tilde\theta_{1}} \Vert \pi_{\varepsilon_{2}}[ \cdot \mid s_{2}(x^\star,  \theta_{1}]) - \pi_{\varepsilon_{2}}\{ \cdot \mid s_{2}(x^\star,  \tilde\theta_{1})\} \Vert_{TV} = \kappa.
\]
The Markov chain produced by Algorithm \ref{AGgen} then  converges geometrically in
total variation distance to a stationary distribution $\nu_{\varepsilon}$, with
geometric rate $1-2\kappa$.
\end{theorem}

The proof of Theorem \ref{THgen} is provided in the Supplementary
Material, Section \ref{proof:THgen} and is based on a coupling argument.

The above assumption is satisfied in particular when the parameter space is compact.
Possible relaxations are not covered in this paper.
This theorem suffers from its generality, as the most practical situation in which the conditions are satisfied is obtained if all the parameters live in a compact space. However we can refine the previous result for many  graphical models; such refinements are explored in  the next sections.

We can extend the convergence result of Theorem \ref{THgen} to the general case $n>2$:

\begin{theorem}
\label{THgengen}
Assume that for all $\ell\leq n$ 
\[ \kappa_\ell =\sup_{\btheta_{>\ell}, \tilde{\btheta}_{>\ell}}\sup_{\btheta_{<\ell}} \Vert \pi_{\varepsilon_\ell}\{ \cdot \mid s_\ell(x^\star, \btheta_{<\ell} ,\btheta_{>\ell})\} -  \pi_{\varepsilon_\ell}\{ \cdot \mid s_\ell (x^\star,\btheta_{<\ell}, \tilde{\btheta}_{>\ell})\}\Vert_{TV} < 1/2 \]
with $\btheta_{>\ell} = (\theta_{\ell+1},\theta_{\ell +2}, \dots, \theta_n)$, and $\btheta_{<\ell} = (\theta_1,\theta_{2}, \dots, \theta_{\ell -1})$.
Then, the Markov chain produced by Algorithm \ref{AGgen}  converges geometrically in
total variation distance to a stationary distribution $\nu_{\varepsilon}$, with
geometric rate $1-\prod_\ell2\kappa_\ell$.
\end{theorem}

The proof of this theorem is a straightforward adaptation of the previous
proof, with the same coupling procedure. The condition comes from the fact that
in this procedure we sequentially try to couple each $\theta_\ell$ using the $\btheta_{<\ell}$, already coupled; as a consequence the condition for $\ell=n$ is always satisfied. In the case $n=2$, we recover Theorem \ref{THgen}.

The limiting distribution $\nu_\varepsilon$ is not necessarily a standard
posterior. We can however provide an evaluation of the distance between
$\nu_\varepsilon$ and the limiting distribution $\nu_0$ of Algorithm
\ref{AGgen} with $\varepsilon_{1} = \varepsilon_{2} = 0$. In a compact parameter
space, $\nu_0$ always exists, but it may differ from the joint distribution
associated with a vanilla ABC sampler, because the conditionals may be based
on different summary statistics $s_{1}$ and $s_{2}$.  

\begin{theorem}
\label{THgenconv}
Assume that
\begin{align*}
L_0 &= \sup_{\varepsilon_{2}} \sup_{\theta_{1},\tilde\theta_{1}} 
\Vert \pi_{\varepsilon_{2}}\{ \cdot \mid  s_{2}(x^\star,\theta_{1})\} - \pi_{0}\{\cdot \mid s_{2}(x^\star,\tilde\theta_{1})\} \Vert_{TV} < 1/2\,, 
\\
L_1(\varepsilon_{1}) &=\sup_{\theta_{2}} \Vert \pi_{\varepsilon_{1}}\{\cdot \mid s_{1}(x^\star,\theta_{2})\} 
- \pi_{0} \{\cdot \mid s_{1}(x^\star, \theta_{2})\} \Vert_{TV} \xrightarrow[\varepsilon_{1} \to 0]{} 0\,, 
\\
L_2(\varepsilon_{2}) &= \sup_{\theta_{1}} \Vert \pi_{\varepsilon_{2}}\{\cdot \mid s_{2}(x^\star,\theta_{1})\}
- \pi_{0} \{\cdot \mid  s_{2}( x^\star,\theta_{1})\} \Vert_{TV} \xrightarrow[\varepsilon_{2} \to 0]{} 0\,.
\end{align*}


Then 
\[ \Vert \nu_\varepsilon - \nu_0 \Vert_{TV} \leq \frac{L_1(\varepsilon_{1}) + L_2(\varepsilon_{2})}{1 - 2L_0} \xrightarrow[\varepsilon \to 0]{} 0. \]
\end{theorem}


\section{Component-wise approximate Bayesian computation: the hierarchical case}
\label{ABCGhierpart}

\subsection{Algorithm and theory}





In this section, we focus on the two-stage simple hierarchical model given in \eqref{eq:modelun}. This model appears naturally when a hierarchical structure is added to a non-tractable model, see for example \cite{turner2013hierarchical}. Under this model structure,
the conditional distributions greatly simplify as 
$\pi(\mu_j     \mid x^\star, \alpha,
\mu_1,\dots,\mu_{j-1},\mu_{j+1},\dots,\mu_{n} ) = \pi(\mu_{j} \mid
x_j^\star,\alpha) $ and $\pi(\alpha \mid \bmu,x^\star)
= \pi(\alpha \mid \bmu) $.
Algorithm \ref{AGgen} then further simplifies and a detailed version in this particular situation is given in Algorithm \ref{algohier}. In order to simulate from all or part of the approximate conditional distributions, we might
resort to a Metropolis step, using the prior distribution as proposal.


\begin{algorithm}[!h]
\SetAlgoLined
\KwIn{observed dataset $x^\star$, number of iterations $N$, starting points
$\alpha^{(0)}$ and $\bmu^{(0)}=(\mu^{(0)}_1,\ldots,\mu^{(0)}_n)$,
thresholds $\varepsilon_\alpha$ and $\varepsilon_\mu$, summary statistics $s_\alpha$ and $s_\mu$, and distances $d_\alpha$ and $d_\mu$.}
\KwOut{A sample $(\alpha^{(i)}, \bmu^{(i)})_{1\le i\le N}$.}
\For{$i = 1 , \ldots, N$}{
\For{$j = 1, \ldots, n$}{
Sample $\mu_j^c \sim \pi(\mu \mid \alpha^{(i-1)})$ and $x_j^c \sim f(\cdot \mid \mu_j^c)$

\While{$d\{s_\mu(x_j^c),s_\mu(x_j^\star)\}>\varepsilon_\mu$}{
Sample $\mu_j^c \sim \pi(\mu \mid \alpha^{(i-1)})$ and $x_j^c \sim f(x_j \mid \mu_j^c)$
}
$ \mu_j^{(i)} \leftarrow \mu_j^c $ 
\tcp*{thus $\mu_j^{(i)}\sim\pi_{\varepsilon_\mu}\{\cdot \mid s_\mu(x_j^{\star},  \alpha^{(i-1)})\}$}
}

Sample $\alpha^c \sim \pi(\alpha)$ and $\bmu^c \sim \pi(\cdot \mid \alpha^c)$,

\While{$d\{s_\alpha(\bmu^c),s_\alpha(\bmu^{(i)})\}>\varepsilon_\alpha$}{
Sample $\alpha^c \sim \pi(\alpha)$ and $\bmu^c \sim \pi(\cdot \mid \alpha^c)$,
}
$ \alpha^{(i)} \leftarrow \alpha^c $ 
\tcp*{thus $\alpha^{(i)}\sim\pi_{\varepsilon_\alpha}\{\cdot  \mid s_\alpha(\bmu^{(i)})\}$}
}
\caption{\ABCG sampler for hierarchical model \eqref{eq:modelun}}
\label{algohier}
\end{algorithm}


As in Algorithm \ref{AGgen}, Algorithm \ref{algohier} may bypass the approximation of some conditionals.
In particular, if $\pi(\alpha \mid \mu)$ can be simulated from and $\pi(\mu
\mid x^\star,\alpha)$ cannot, we prove in the Supplementary Material, Section \ref{specificproofs}, that the
limiting distribution of our algorithm is the same as the vanilla \ABC~algorithm. On the other hand,
if we can simulate from $\pi(\mu \mid x^\star,\alpha)$ and not from $\pi(\alpha \mid \mu)$,
a version of Theorem \ref{THgengen} (Theorem \ref{thm:k1k2k3}) is established under less stringent conditions in
the Supplementary Material, Section \ref{specificproofs}.

\subsection{Numerical comparison with vanilla \ABCS}
\label{modeltoy}

We now compare the \ABCG, Vanilla ABC and an implementation of the SMC-ABC algorithm (approximate Bayesian computation with sequential Monte Carlo) of \cite{delmoral:doucet:jasra:2012}, with an adaptive proposal and resampling steps,
following \cite{toni2008} in order to avoid degeneracy in the simulation. The example is the toy Normal--Normal model from \cite{gelman2013bayesian}:
\begin{equation}\label{eq:gelman}
\mu_{j} \stackrel{\text{iid}}{\sim} \mathcal{N}(\alpha,\varsigma^2)\,,
\quad x_{j,k} \stackrel{\text{ind}}{\sim} \mathcal{N}(\mu_{j},\sigma^2)\,,
\qquad j =1, \dots, n,\quad  k=1,\ldots, K
\end{equation}
with the variances $\sigma^2$ and $\varsigma^2$ known,
and a hyperprior $\alpha\sim\mathcal{U}[-4,4]$. The assumptions of Theorem \ref{THgengen} hold here, as shown in the supplementary material, Section \ref{assumptioncheck}. This model is not intractable, which allows us to compare the output with the true posterior in Figure \ref{superp}.

Recall that in practice the tolerance is provided by an empirical quantile of the distance distribution at each call of an approximate conditional. This means that at each iteration $N_\alpha$ and $N_\mu$ simulations are produced from the conditional prior predictives on $\alpha$ and $\mu$, respectively, and that only the simulation associated with the smallest distance is kept.  In Section \ref{SM} we explore some further variations on this implementation. The \texttt{R} code used for all simulations can be found at \texttt{https://github.com/GClarte/ABCG}.

%

We strive to provide a fair comparison between \ABCG\ and vanilla \ABCS\ and
hence aim at simulating overall the same number of normal random variables. In
\ABCS, simulating over the hierarchical structures involves $n+nK$ normal
variates; taking the best $N$ out of $N_V$ prior predictive
simulations thus costs
$\Ntot=N_Vn(1+K)$. In \ABCG, each iteration costs $N_\alpha n + N_\mu nK$; if
we take $N_\alpha=N_\mu$ the total cost is $\Ntot=NnN_\alpha (1+K)$. We thus
take $N=N_V/N_\alpha$ to compare both algorithms.

Figure \ref{fig1} illustrates the result of both algorithms, for $\sigma=1$,
$K=10$, $n=20$, by representing the posterior approximation from \ABCG and
Vanilla ABC for the hyperparameter $\alpha$ and the first three parameters,
with comparable computational costs. The statistic used at both parameter and
hyperparameter levels is the corresponding empirical mean and hence is
sufficient. We keep $N$ constant and increase $N_\alpha = N_\mu$.

This toy experiment exhibits a considerable improvement in the parameter
estimator when using \ABCG. This is easily explained by the difficulty for \ABCS \ to
find a suitable value of $\bmu\in\mathbb{R}^{20}$; poor estimation of the
parameter ensues. In fact,  \ABCS \ produces the same output as a
non-hierachical model when the $\mu_{j}$'s are integrated out. 

This figure exhibits that \ABCG scales more efficiently with $N_\mu$, that is
with the reduction of $\varepsilon_\mu$, especially when increasing $N_\alpha$
from $5$ to $30$, that is a mere $6$ time increase in the computational cost.
For the same variation in \ABCS, we do see no noticeable improvement. Hence,
for a given computational cost, ABC-Gibbs achieves a smaller threshold
$\epsilon$ than ABC, leading to better approximations. The experiment further
points out that the choice of the parameters $N_\alpha$ and $N_\mu$ may prove
delicate. Resorting to a larger ABC table for each update is uselessly costly in
that it fails to provide a clear improvement in the result. This is also the case with the
classical ABC approach, as shown by Figure \ref{fig1}.

In practice, the choice of the parameters $N$ and $N_\mu$ may be tricky. For $N$ we would advise to use standard techniques to choose the number of iterations in a Monte Carlo algorithm. For $N_\mu$ and $N_\alpha$, we observe in Figure \ref{fig1} that a moderate value, say $N_\alpha = N_\mu= 30$, seems enough: we expect the optimal value to be problem dependent.


\begin{figure}
\center
\includegraphics[scale=.9]{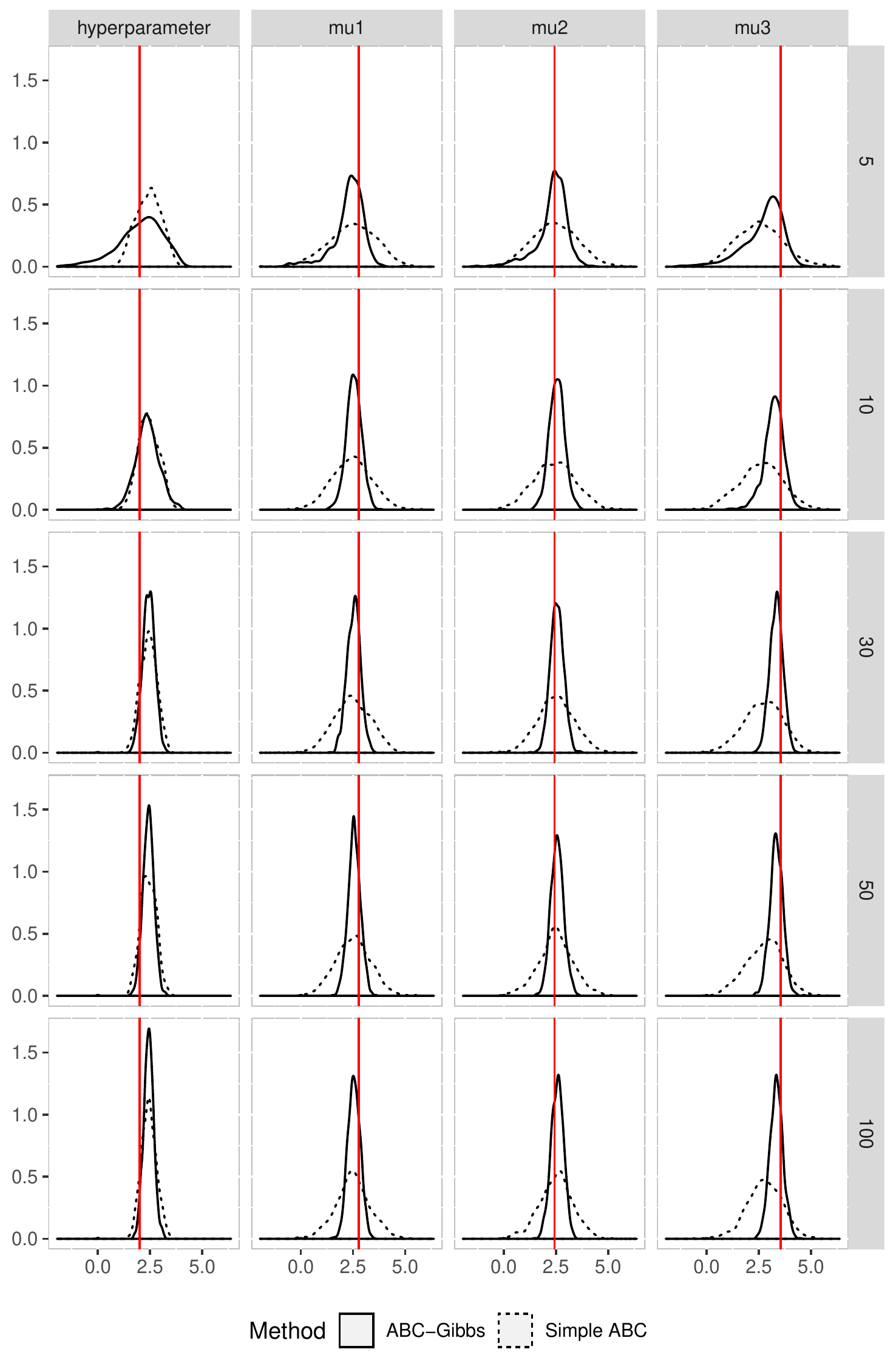} 
\caption{Comparison of the posterior density estimates of the hyperparameter and the first three parameter components of the hierarchical model of Equation \ref{eq:gelman} obtained with ABC and \ABCG, with identical computational cost. For \ABCG, results were computed with $N=1000$ Gibbs iterations;  each row is labeled with the number $N_{\alpha}=N_\mu$ of iterations of the ABC scheme to update one parameter component. Red vertical lines represents the true values used to simulate the data.}
\label{fig1}
\end{figure}

To check the robustness of our method, we represent in Figure \ref{superp}, 10
realisations of the posterior densities, for $N = \lfloor 1000/30 \rfloor$, and
$N_\alpha = N_\mu = 30$ (with the first $5$ points in \ABCG removed to account
for the small burn in). SMC-ABC does not allow for a fixed limit on the
number of simulations, due to the resampling step. We used therefore $10^4$ particles,
with a target of the smallest possible tolerance for a maximum
of $30$ steps. In total, SMC-ABC was alloted roughly 60 times more
simulations than \ABCG and \ABCS. The \ABCG \ density is overdispersed
compared to the true posterior, albeit closer than the \ABCS, especially for
the parameter $\mu_{1}$. On the other hand, SMC-ABC fails for this model:
due to the difficulties resulting from its high dimension, an adaptive version
fails to produce interesting proposals, notwithstanding a consistently larger
computational budget. The distribution approximation on $\alpha$ is however
better than its \ABCS counterpart. This fact is supported by numerical
experiments in lower dimensions where all three methods lead to suitable
approximations, as illustrated in the supplementary material, Section
\ref{graphiques3D}. The improvement brought by \ABCG \ in high dimension occurs consistently 
over simulations. 

This experiment further highlights a striking differenciation between ABC-SMC methods, which require a significant degree of calibration when no package is readily available, and \ABCG, which relies on a straightforward implementation, reproduced in the supplementary material.


\begin{figure}
\center
\includegraphics[scale=.7]{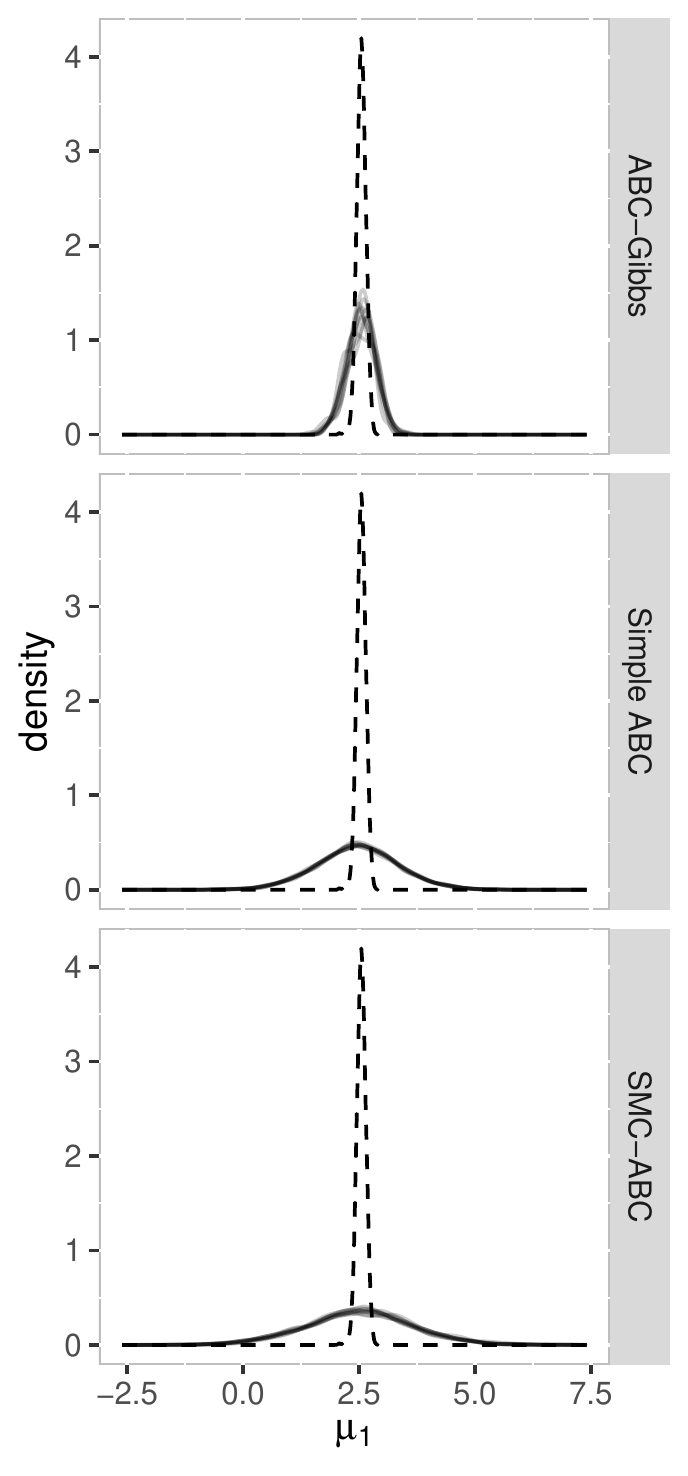} 
\includegraphics[scale=.7]{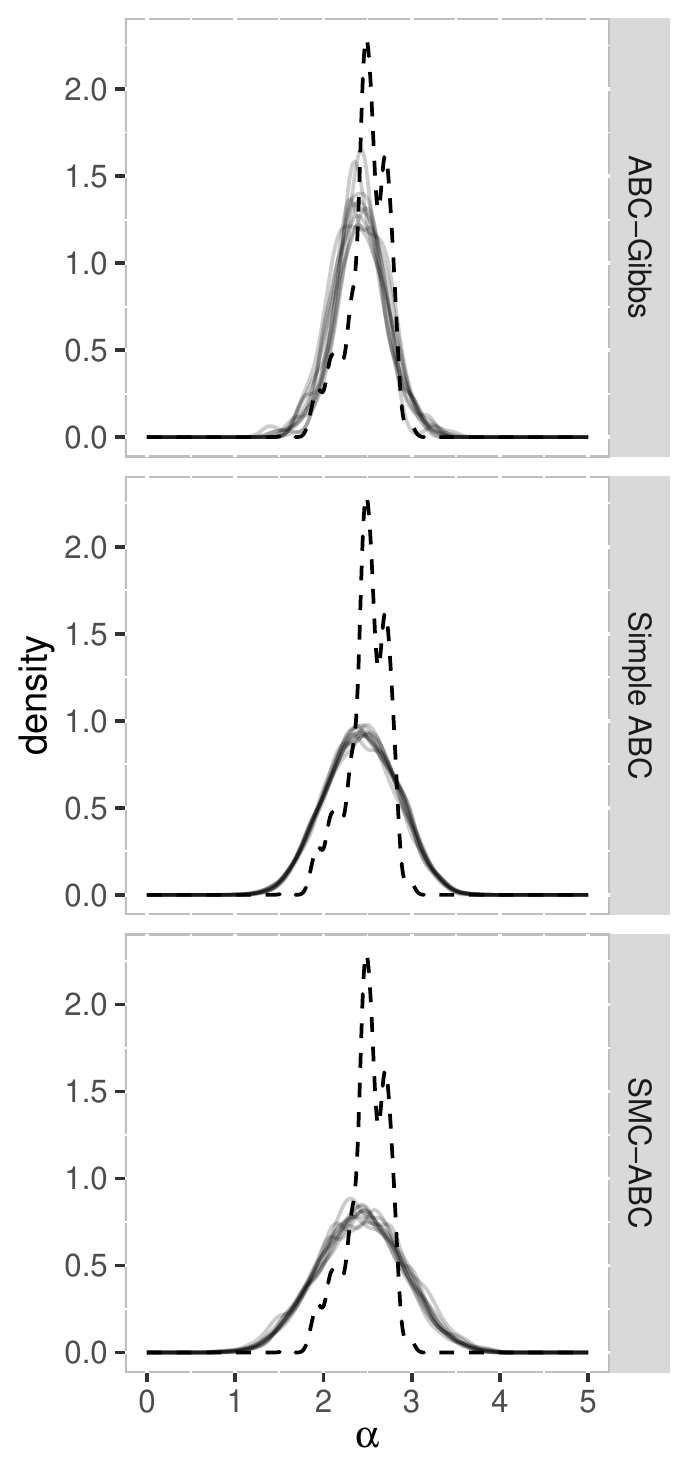} 
\caption{Posterior densities for 10 replicas of the algorithms compared to the exact posterior density. The true posterior is represented by the dashed line.}
\label{superp}
\end{figure}



\section{Application: hierarchical G \& K distribution}

The G \& K distribution is a notorious example \citep{prangle2017gk} of an intractable distribution.
It depends on parameters $(\mu, B, g, k)$ and is defined by its inverse
cumulative distribution function
\[
F^{-1}(x;\mu,B,g,k,c) = \mu+B\left(1+c\frac{1-e^{-gz(x)}}{1+e^{-gz(x)}}\right)\left(1+z(x)^2\right)^kz(x)
\]
where $z$ is the quantile function of the standard normal distribution, and $c$
is a constant typically set to $0.8$ \citep{prangle2017gk}. While the
likelihood function is intractable, it is straightforward to simulate
realisations of this distribution, making it a favourite benchmark for ABC
methods (see for instance \citealp{fearnhead2012}).

Here, we analyse two hierarchical versions of this model, both of the form: \begin{equation}
 \mu_i \sim \mathcal{N}(\alpha,1) \qquad x_i\sim gk(\mu_i,B,g,k)\qquad i=1,\ldots,n\,.
\label{eq:gk}
 \end{equation}

In a first experiment, we assume that the parameters $B$, $g$ and $k$ are known and  we infer the position parameters $(\mu_i)$. This leads to the graphical model represented on the left of Figure \ref{GK}. We refer to this model as the simple hierarchical G \& K model.

\begin{figure}
\center
\begin{tikzpicture}

\node (a) at (-2.4, 1.5) {$\alpha$};
\node (b1) at (-1.2, 3) {$\mu_{1}$};
\node (b2) at (-1.2, 2)  {$\mu_{2}$};
\node at (-1.2, 1) {$\vdots$};
\node (b3) at (-1.2, 0) {$\mu_{n}$};

\draw [->] (a)--(b1);
\draw [->] (a)--(b2);
\draw [->] (a)--(b3);

\node (x1) at (0, 3) {$x_{1}$};
\node (x2) at (0, 2) {$x_{2}$};
\node (x3) at (0, 0) {$x_{n}$};
\node at (0, 1) {$\vdots$};

\draw [->] (b1)--(x1);
\draw [->] (b2)--(x2);
\draw [->] (b3)--(x3);

\end{tikzpicture} \hspace{2cm}\begin{tikzpicture}

\node (a) at (-2.4, 1.5) {$\alpha$};
\node (b1) at (-1.2, 3) {$\mu_{1}$};
\node (b2) at (-1.2, 2)  {$\mu_{2}$};
\node at (-1.2, 1) {$\vdots$};
\node (b3) at (-1.2, 0) {$\mu_{n}$};

\draw [->] (a)--(b1);
\draw [->] (a)--(b2);
\draw [->] (a)--(b3);

\node (x1) at (0, 3) {$x_{1}$};
\node (x2) at (0, 2) {$x_{2}$};
\node (x3) at (0, 0) {$x_{n}$};
\node at (0, 1) {$\vdots$};

\draw [->] (b1)--(x1);
\draw [->] (b2)--(x2);
\draw [->] (b3)--(x3);

\node (s1) at (2, 2.4) {$B$};
\node (s2) at (2, 1.5) {$g$};
\node (s3) at (2, 0.6) {$k$};

\draw [->] (s1)--(x1);
\draw [->] (s2)--(x1);
\draw [->] (s3)--(x1);

\draw [->] (s1)--(x2);
\draw [->] (s2)--(x2);
\draw [->] (s3)--(x2);

\draw [->] (s1)--(x3);
\draw [->] (s2)--(x3);
\draw [->] (s3)--(x3);

\end{tikzpicture}

\caption{Left: Simple hierarchical G \& K model; Right: doubly hierarchical G \& K  model}
\label{GK}
\end{figure}
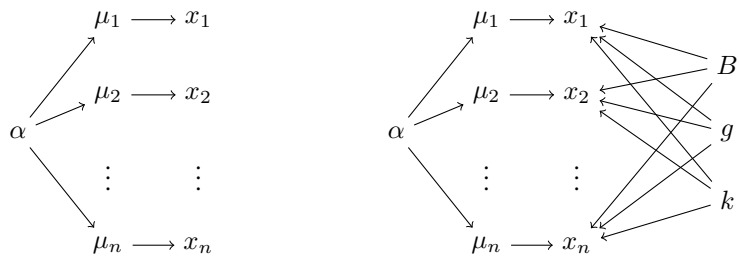

For a hyperprior $\alpha\sim\mathcal U[-10,10]$, the assumptions of  Theorem $\ref{thm:k1k2k3}$ are satisfied. Figure \ref{compsimpleGK} compares the results of our algorithm with those of \ABCS for a similar computational cost in dimension $n=50$, and ABC-SMC (same as before) for a higher computational cost, with $1000$ particles, $500$ iterations leading to a computational cost roughly $20$ times longer.  As in Section \ref{modeltoy}, \ABCG outperforms both other methods: Vanilla ABC is overdispersed, and carefully calibrated ABC-SMC is either highly peaked at the wrong location or producing results similar with \ABCS.

\begin{figure}
\center
\includegraphics[scale=1]{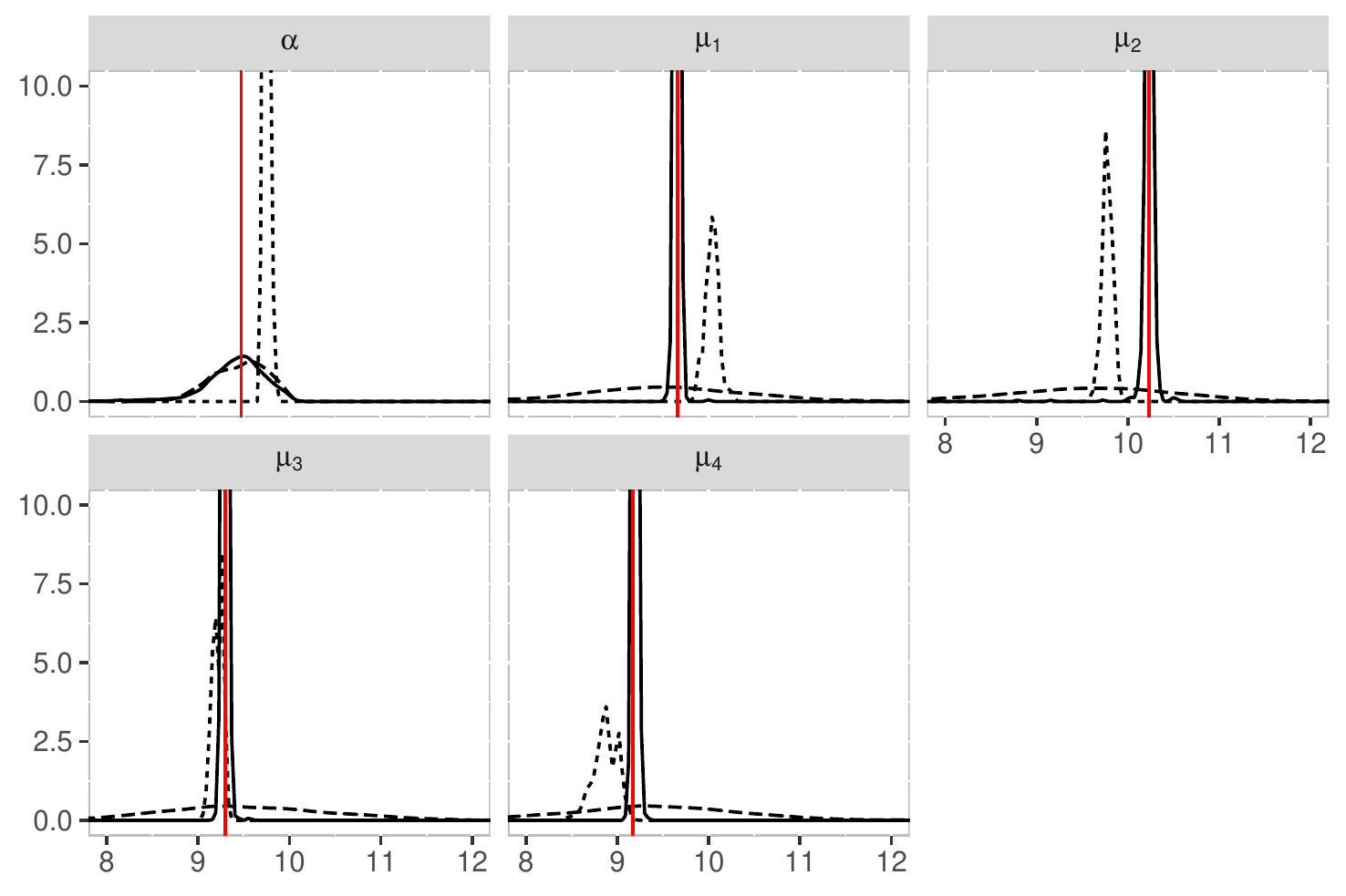}
\caption{Posterior approximations for the simple hierarchical G \& K model. The $y$ axis is truncated as the ABC-SMC pseudo-posterior is very peaked. The red vertical lines identify the value of the parameters used in the simulation.}
\label{compsimpleGK}
\end{figure}

As a second experiment, we infer 
all parameters ($B$, $g$, $k$, $\alpha$ and the $\mu_i$) in Equation \ref{eq:gk}, with independent hyperpriors $\alpha\sim\mathcal U(-10,10)$ and $B,g,k\sim\mathcal U(0,1)$. This corresponds to the graphical model  represented on the right of Figure \ref{GK}, which we refer to as a doubly hierarchical G \& K model. 
The same summary statistic is used at every step of the algorithm, namely the octiles of the observations. 
Let $q(x,p)$ be the $p$-th quantile of sample $x$ and take two observations $x_1$ and $x_2$; our distance function is
\[ d(x_1,x_2) = \sum_{i=0}^8 \left\vert q\left(x_1,i/8\right) - q\left(x_2,i/8\right) \right\vert. \]
It is straightforward to  that the assumptions of Theorem \ref{THgengen} are satisfied by this model, when considering the parameters in the order $\alpha, B, g, k, (\mu_i)$.

Figures \ref{compsimpleGK}, \ref{fig:compsimpleGKpar} and
\ref{fig:compsimpleGKhyp} compare the output of ABC-Gibbs with Vanilla ABC and
ABC-SMC in the same implementation as before, under a fixed budget of
$2\cdot 10^6$ model simulations for \ABCG and \ABCS; ABC-SMC is run with
$10^3$ particles, for $500$ steps, resulting in a a grand total computational cost 
larger than $2.5\cdot 10^7$ simulations. Note that there exist
analytical approximations of the G \& K posterior that give better results than
ABC methods in the non hierarchical case, however none of these methods can
be easily extended to the hierarchical case. 

\begin{figure}
\centering
\includegraphics[scale=1]{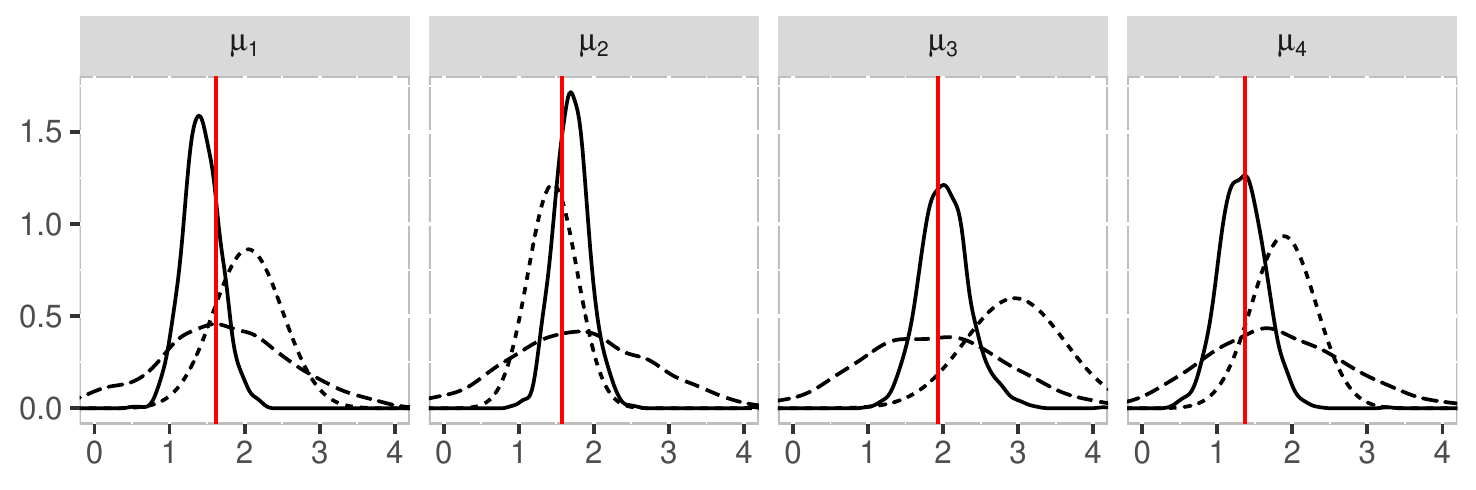} 
\caption{Posterior densities for the first four parameters, among 50, $\mu_1, \dots , \mu_4$ in the doubly hierarchical $g$ \& $k$ model.}
\label{fig:compsimpleGKpar}
\end{figure}

\begin{figure}
\centering
\includegraphics[scale=1]{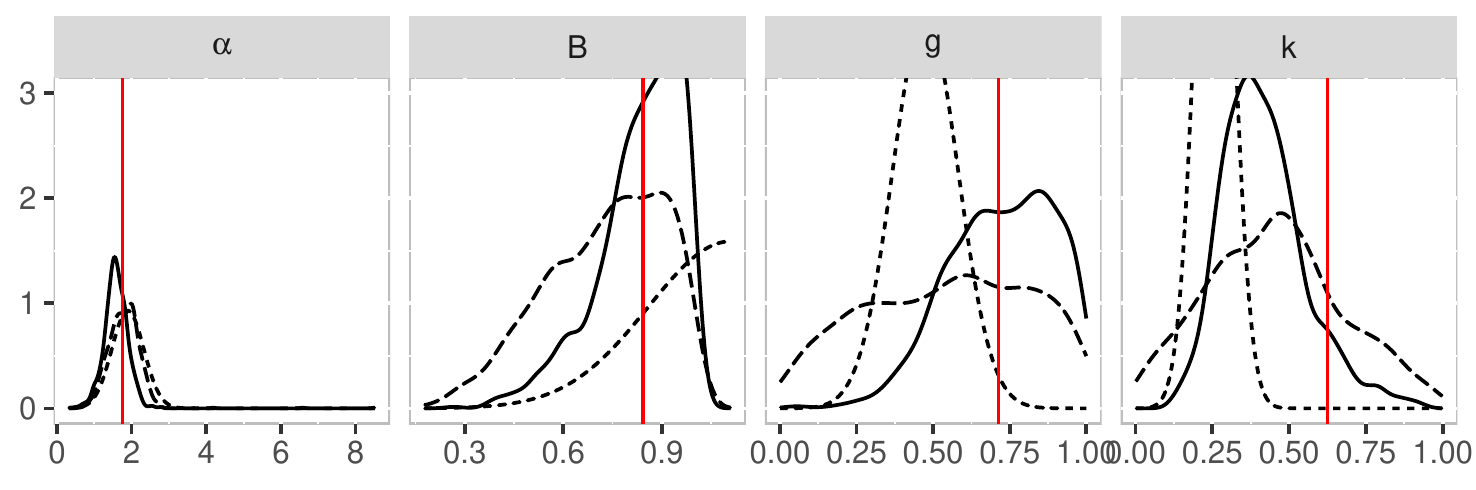} 
\caption{Posterior densities for the top-level parameters $\alpha$, $B$, $g$ and $k$ in the doubly hierarchical $g$ \& $k$ model}
\label{fig:compsimpleGKhyp}
\end{figure}

The simple and double hierarchical G \& K models lead to comparable results.
\ABCG provides consistently better results (that are more concentrated around the
true value), than \ABCS. The approximation provided by ABC-SMC  is less
peaked and occasionaly exhibits a bias, if less visible in the
hyperaparameter case. Sequential Monte Carlo is supposed to iteratively reduce the
threshold of the approximation; however, due to the difficulty of calibrating, the starting 
points, the reduction is quite slow. It is thus unlikely a further increase in
the computational time would lead to higher improvements.

In Section \ref{sec:MA} of the Supplementary material, we consider another example (a hierarchical Moving Average model), for which the results are similar.

%

\section{Example with full dependence}
\label{Nonhier}

The concept of \ABCG is by no means restricted to hierarchical settings. It applies in full generality to any decomposition or completion of the parameter $\btheta$ into $n$ terms, $(\theta_1,\ldots,\theta_n)$. For simplicity's sake, we only analyse below the case of $n=2$ parameters, and furthermore assume that $\theta_1$ and $\theta_2$ are a priori independent. The
extension to $n \geq 2$ parameters, or non-independent parameters, is straightforward though cumbersome.
The generic Algorithm \ref{AGgen} and Theorem \ref{THgen} can thus be adapted to non-hierarchical models where $\btheta=(\theta_1,\theta_2)$, 
such that the conditional posteriors $\pi(\theta_1\mid x^\star,\theta_2)$ and
$\pi(\theta_2\mid x^\star,\theta_1)$ depend on the entire dataset $x^\star$ rather than a significantly smaller
subset. This setting
implies that the approximation steps in \ABCG will mostly require the simulation
of objects of the same size as in \ABCS. 
%
%

When the statistics $s_1$ and $s_2$ are identical, a single distance can be
used, with $\varepsilon_1=\varepsilon_2$. The resulting stationary
distribution is then the same as for \ABCS, since it is proportional to 
$$ \int \pi(\theta_1) \pi(\theta_2) f(x \mid \theta_1,\theta_2) \mathbf 1_{\eta\{s_1(x),s_1(x^\star)\}\le \varepsilon_1}\, \mathrm{d}x.$$
Formally, these statistics should however be different, since more efficient and
smaller-dimension statistics can be calibrated to each parameter. 

As an illustration, consider an example inspired by inverse problems
\citep{Kaipio:2011}, in a simplified version. These problems, although deterministic, are difficult to tackle with traditional methods, as the likelihood function is typically extremely expensive to compute \citep{neal2012efficient}, requiring the use of surrogate models, and thus approximations. Let $y$ be the solution
of the heat equation on a circle defined for $(\tau, z) \in ]0,T]\times[0,1[$ by
\[  \partial_\tau y(z,\tau) = \partial_z \left\{ \theta(z) \partial_z y(z,\tau)
\right\},  \]
with $\theta(z) = \sum_{j=1}^n \theta_j \mathbf{1}_{\{(j-1)/n,j/n\}}(z)$ and with
boundary conditions  $ y(z,0)=y_0(z)$ and  $y(0,\tau)=y(1,\tau)$.
We assume $y_0$ known and the parameter is $ \theta
=(\theta_1, \dots, \theta_n)$. The equation is discretized towards its
numerical resolution. For this purpose, the first order finite elements method
relies on discretisation steps of size $1/n$ for $z$ and $\Delta$ for $\tau$.
A stepwise approximation of the solution is thus $\hat{y}(z,t) = \sum_{j=1}^n
y_{j,t} \phi_j(z)$, where, for $j<n$, $\phi_j(z) = (1 - \vert
nz-j \vert) \mathbf{1}_{\vert nz-j \vert < 1} $
and $\phi_n(z)=(1-nz)\mathbf{1}_{0<z<1/n}+(nz-n+1)\mathbf{1}_{1-1/n < z < 1}$, and with $y_{j,t}$ defined by
\begin{align*}\frac{y_{j,t+1}-y_{j,t}}{3\Delta}+\frac{y_{j+1,t+1}-y_{j+1,t}}{6\Delta}&+\frac{y_{j-1,t+1}-y_{j-1,t}}{6\Delta}\\
&\quad = y_{j,t+1}(\theta_{j+1}+\theta_{j})-y_{j-1,t+1}\theta_{j}-y_{j+1,t+1}\theta_{j+1}. \end{align*}
We then observe a noisy version of this process, chosen as $x_{j,t} = \mathcal{N}(\hat{y}_{j,t},\sigma^2)$. 

In \ABCG, each parameter $\theta_m$ is updated with summary statistics
the observations at locations $m-2,m-1,m,m+1$. \ABCS relies on the whole data as statistic. 
In the experiments, $n=20$ and $\Delta=0.1$, with a prior $\theta_j \sim\mathcal U [0,1]$, independently. 
Theorem \ref{THgengen} applies to this setting.

We compared both methods, using as above the same simulation budget and
several experiments with various values of
$N_\epsilon$, keeping the total number of simulations constant at
$\Ntot=N_\epsilon\cdot N=8\cdot 10^6$. As $N_\epsilon$ increases, the size $N$
of the posterior sample decreases. Figure \ref{resnonhier} illustrates
the estimations of $\theta_1$.  The \ABCG estimate is much closer to the
true value of the parameter $\theta_1=0.75$, with a smaller variance.  We
emphasize once more that the choice of the ABC table size is critical, as for a
fixed computational budget we must reach a balance between, on the one hand,
the quality of the approximations of the conditionals (improved by increasing
$N_\epsilon$), and on the other hand Monte-Carlo error and convergence of the
algorithm, (improved by increasing $N$). In our case, $N_\varepsilon = 10$ was
clearly the best choice (low bias and low variance). While we have no
systematic rule to choose this parameter, we however advise to choose it
so that the approximation of the conditional is significantly
different from the prior when run separately.


As  in previous instances, \ABCG is much more efficient than \ABCS. For
instance, Figure \ref{resnonhier2} shows that the posterior sample of
$\theta_1$ is more peaked around the true value for \ABCG.  We  repeated this
experiment for a wide range of values for $\theta$. In all, \ABCG gives
estimates close to the true value, and is never outperformed by \ABCS. This is
confirmed by evaluating the expectation of the posterior predictive distance to
the whole data, \ABCG achieves an average of $39.2 \pm 0.002$, and
\ABCS reaches an average of $103.8 \pm 0.002$, based on 100 replicates.

\begin{figure}
\center
\includegraphics[scale=.5]{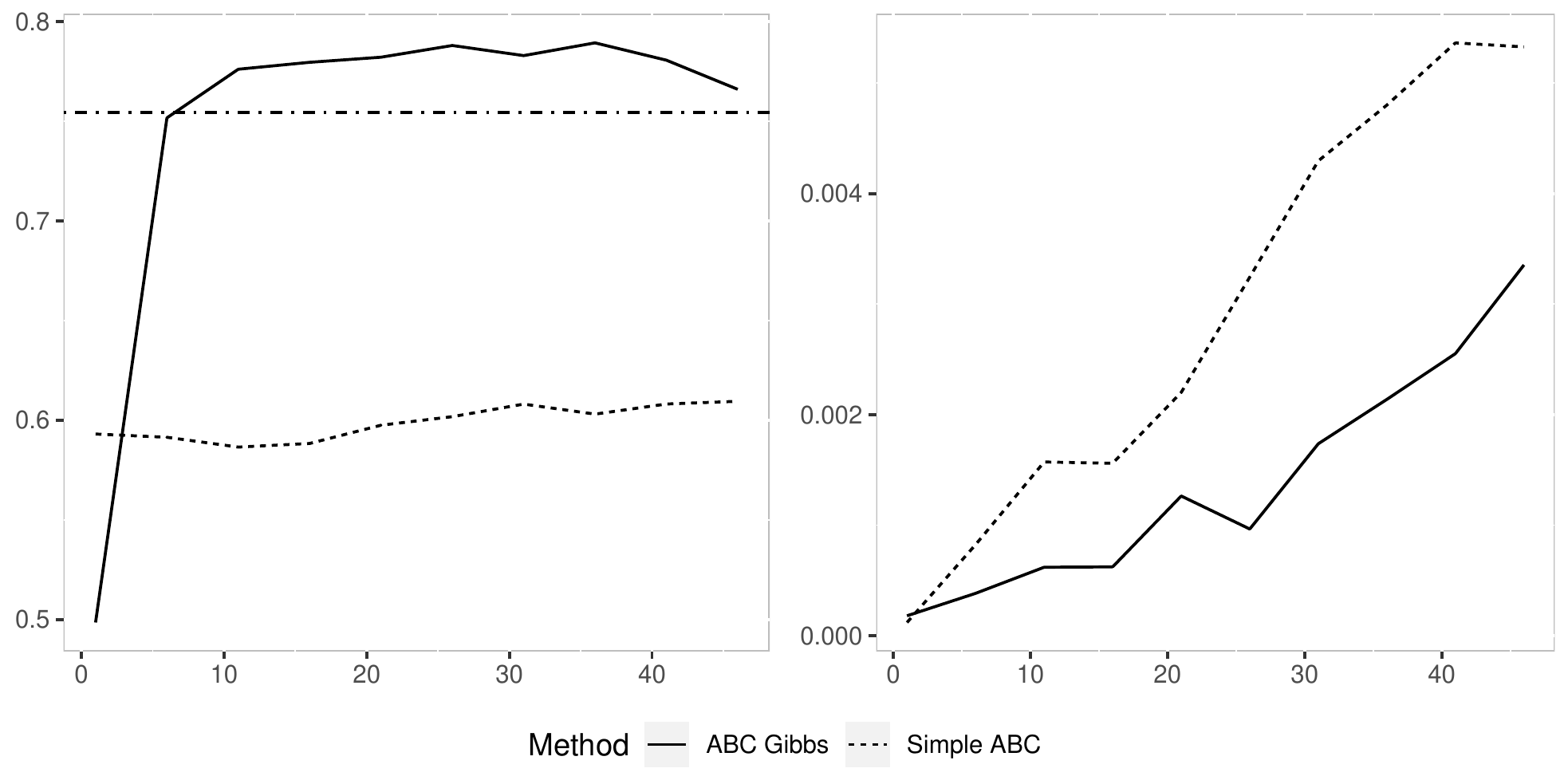} 
\caption{For the heat equation model, mean and variance of the ABC and \ABCG estimators of $\theta_1$ as $N_\epsilon$ increases, selected from among 20 parameters. The horizontal line shows the true value of $\theta_1$.}
\label{resnonhier}
\end{figure}

\begin{figure}
\center
\includegraphics[scale=.5]{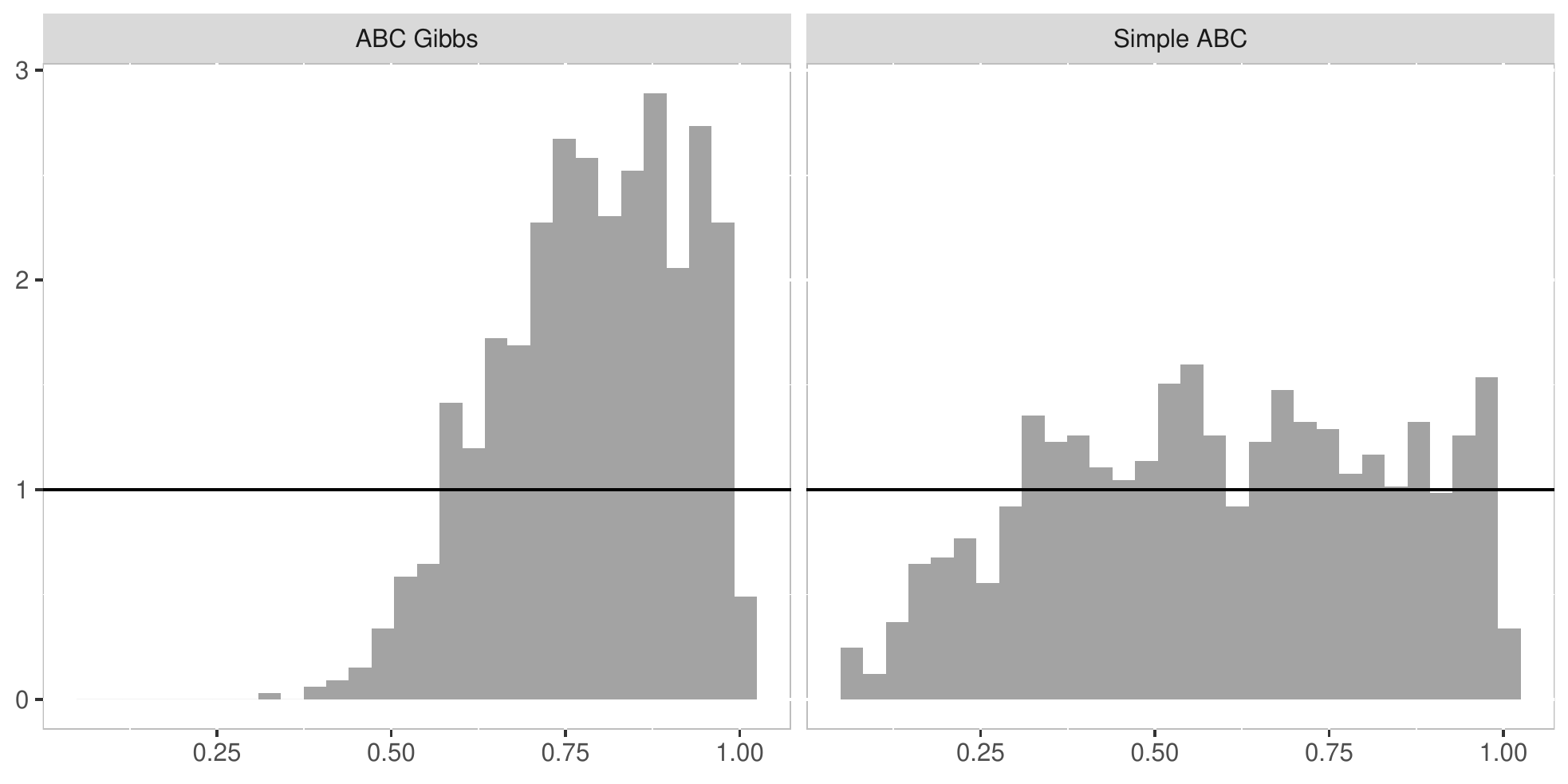} 
\caption{For the model of section \ref{Nonhier}, approximate posterior of $\theta_1$ compared with the uniform prior (black line) for \ABCG\ (right) and \ABCS\ (left)}
\label{resnonhier2}
\end{figure}

\section{Nature of the limiting distribution}
\label{sec:comp}
\label{incomp1}


When  addressing hierarchical models of the form of Equation \eqref{eq:modelun}, 
 we gave conditions in Theorem \ref{THgengen} for Algorithm \ref{AGgen} 
to have a limiting distribution $\nu_\epsilon$.  However, we did not specify the nature of this limiting distribution. We also showed that as the
tolerance parameter $\epsilon$ goes to 0, $\nu_\epsilon$ tends to the
stationary distribution $\nu_0$ of a Gibbs sampler with generators $\pi(\alpha)
\pi(s_\alpha(\mu) \mid \alpha)$ and $\pi(\mu \mid \alpha) \pi(s_\mu(x) \mid \mu)$. 
It is possible that these generators are incompatible, that is, that there is
no joint distribution associated with them. In such settings, the stationary
distribution $\nu_0$ does not enjoy these generators as conditionals.
The incompatibility of conditionals may seem contradictory with the fact that
our algorithm does converge to a distribution, but in the case of a compact parameter space there
always exists a limiting distribution, the main issue being rather that the
limiting distribution has no straightforward Bayesian interpretation.

There are however specific situations where there are theoretical guarantees that the limiting distribution $\nu_0$ is in fact the posterior distribution associated with the summary statistics.
According to
\cite{arnold1989compatible} a necessary and sufficient condition for the
conditionals to be compatible is the existence of  two measurable
functions $u(\alpha)$ and $v(\mu)$ such that 
\[\frac{\pi(\alpha)
\pi(s_\alpha(\mu) \mid \alpha)}{ \pi(\mu \mid \alpha) f(s_\mu(x) \mid
\mu)}=u(\alpha)v(\mu).
\]
In particular, this occurs if $s_\alpha$ is sufficient. (This condition is not necessary, as it is also true for example if $s_\alpha$ is ancillary, although this is of limited interest.)

We have thus proven the following Proposition,

\begin{proposition}

Under the assumptions of Theorems \ref{THgen} and \ref{THgenconv}, a limiting
distribution exists and converges, for both $\varepsilon_\mu$ and
$\varepsilon_\alpha$ decreasing to $0$, to the stationary distribution of a Gibbs
sampler with conditionals:
\[\pi(\alpha) \pi(s_\alpha(\mu) \mid \alpha) \text{ and } \pi(\mu \mid \alpha) f(s_\mu(x) \mid
mu).\]
If $s_\alpha$ is sufficient, this limiting distribution is merely $ \pi(\alpha,\mu) \pi(s_\mu(x) \mid \mu) $, that is the limiting distribution of \ABCS with summary statistic $s_\mu$ when the tolerance goes to $0$.

\end{proposition}

We can state similar results for non hierarchical models, although each model requires its own proof. For example, for the full dependency model \ref{Nonhier} with two parameters $\theta_1$ and $\theta_2$ we have the following Proposition:

\begin{proposition}

If $\pi(\theta_1,\theta_2)$ = $\pi(\theta_1)\pi(\theta_2)$ and $s_{\theta_1} = s_{\theta_2}$, as the tolerance $\varepsilon$ goes to zero and under the assumptions of Theorems \ref{THgengen} and \ref{THgenconv}, \ABCG and \ABCS have the same limiting distribution.

\end{proposition}

\section{Discussion}
\label{sec:dis}

The curse of dimensionality remains the major jamming block for the expansion
of ABC~methodology to more complex models as most of its avatars see their cost
rise with the dimensions of the parameter and of the data \citep{li2018}. This
is particularly the case for high-dimensional parameters, since they require
summary statistics that are at least of the same dimension and, unless the
model under study is amenable to easily computed estimates of these parameters,
a much larger collection of statistics is usually unavoidable. Breaking this
curse of dimensionality by Gibbs-like steps is thus as important for
ABC~methods as it was for Monte Carlo methods \citep{GelflandSmith}, as relying
on a small number of summary statistics facilitates the derivation of automated
or semi-automated approaches \citep{fearnhead2012,ABCRF} and offers the
potential for simulating pseudo-data of much smaller size. In appropriate
settings, \ABCG sampling provides a noticeable improvement of the efficiency of
approximate Bayesian computation methods. We have established some sufficient
conditions for the convergence of \ABCG~algorithms. Questions remain, from
checking such conditions in practice to a better understanding of the limiting
distributions from an inferential viewpoint. 
A Gibbs-like setting could also allow practitioners to embed their model in a higher-dimensional model with auxiliary variables,
with compatible conditionals and improved computational tractability.
In all cases, constructing or selecting a low-dimension
informative
summary statistic for the approximation of the conditionals might be an
unavoidable challenge to further improve the quality of the results. 

\section*{Acknowledgements}

This paper greatly benefited from early discussions with Anthony Ebert, Kerrie Mengersen, and Pierre Pudlo, as well as detailed and helpful suggestions from an anonymous reviewer, to whom we are most grateful. We also acknowledge the Jean Morlet Chair which partly supported meeting in the Centre International de Rencontres Math\'ematiques in Luminy. The second author is also affiliated with the Department of Statistics, University of Warwick.  He was supported in part the French government under management of Agence Nationale de la Recherche as part of the ``Investissements d’avenir'' program, reference ANR19-P3IA-0001 (PRAIRIE 3IA Institute).


\newpage
\setcounter{page}{1}
\setcounter{section}{7}
\setcounter{algocf}{4}
\setcounter{figure}{8}
\setcounter{theorem}{3}
\SetKwRepeat{Do}{do}{while}

\centerline{\bf Component-wise Approximate Bayesian Computation via Gibbs-like steps: Supplementary Material}

\section{Supplementary material for Section \ref{modeltoy}}
\label{SM}

\subsection{Checking the assumptions of Theorem \ref{THgenconv}}

\label{assumptioncheck}

In this section, we show that the assumptions of Theorem \ref{THgenconv} apply to the toy model of Equation \ref{eq:gelman} in Section \ref{modeltoy}.

We define $\mu_{-i} = (\mu_1,\dots , \mu_{i-1},\mu_{i+1}, \dots ,\mu_n) $. By conditional independence of the $\mu_i$ given $\alpha$, and choice of $s_\mu$, we have:

\[ \pi_{\varepsilon_\mu}\{ \cdot \mid s_\mu (x^\star, \alpha, \mu_{-i})\} = \pi_{\varepsilon_\mu}\{ \cdot \mid s_\mu(x^\star_i,\alpha)\}\]

The assumptions to check can be rewritten as:

\begin{enumerate}
\item[$(\mu_1)$]$\sup_{\alpha, \tilde{\alpha}} \Vert \pi_{\varepsilon_\mu}\{ \cdot \mid s_\mu(x^\star_1,\alpha) \} -  \pi_{\varepsilon_\mu}\{ \cdot \mid s_\mu(x^\star_1,\tilde{\alpha})\}\Vert_{TV} < 1/2$
\item[$\vdots$] 
\item[$(\mu_n)$]$\sup_{\alpha, \tilde{\alpha}} \Vert \pi_{\varepsilon_\mu}\{ \cdot \mid s_\mu(x^\star_n,\alpha) \} -  \pi_{\varepsilon_\mu}\{ \cdot \mid s_\mu(x^\star_n),\tilde{\alpha}\}\Vert_{TV} < 1/2$
\item[$(\alpha)$] $\sup_{\bmu}\Vert \pi_{\varepsilon_\alpha}\{ \cdot \mid s_\alpha(\bmu) \} -  \pi_{\varepsilon_\alpha}\{ \cdot \mid s_\alpha (\bmu)\}\Vert_{TV} < 1/2$
\end{enumerate}

To prove the assumption $(\mu_i)$, we underline the fact that it is sufficient to check that there exists some subset $K$ of the parameter space, with positive measure for all hyperparameter $\alpha$, such that $ \exists C >0 , \forall \alpha, \forall \mu \in K$,  $\pi_{\epsilon_\mu}( \mu \mid \alpha,s(x^*) )>C$. 

We can compute these densities:

\[ \pi_\varepsilon( \mu \mid \alpha,s(x^*) ) = \frac{\exp\{-(\mu - \alpha)^2/(2\tau) \int \exp(-(y-\mu)^2\sqrt{n}/(2\sigma)\} \mathbf{1}_{\vert y - \bar{x^\star} \vert < \varepsilon} \mathrm{d}y}{ \int\exp\{-(\mu - \alpha)^2/(2\tau) \exp(-(y-\mu)^2\sqrt{n}/(2\sigma) \} \mathbf{1}_{\vert y - \bar{x^\star} \vert < \varepsilon} \mathrm{d}y\mathrm{d}\mu}. \]

As $\alpha$ is compactly supported on $[-4,4]$, the conditions are verified: we can roughly bound the probabilities by continuity of the expression. 

The last condition $(\alpha)$ is always verified as we have by definition of the total variation distance:

\[  \sup_{\bmu} \Vert \pi_{\varepsilon_\alpha}(\cdot \mid \bmu) - \pi_{\varepsilon_\alpha} ( \cdot \mid \bmu) \Vert_{TV} = 0. \]


\subsection{Comparison in dimension 3}

\label{graphiques3D}

In addition to the results shown in Figure \ref{superp}, we show in Figure \ref{fig:toy1D} a comparison of \ABCG, Vanilla ABC and SMC-ABC for the toy model of Section \ref{modeltoy} in the low dimension case $n=2$.

As expected, in this low-dimension setting the results from SMC-ABC and Vanilla ABC are comparable to the approximate posterior provided by ABC-Gibbs for the parameter. \ABCG however seems to lead to a less stable approximation of the hyperparameter, this can be explained by the lower number of points (as we removed some of the first points as burn-in). This supports the idea that the behaviour of SMC-ABC in Figure \ref{superp} is caused by the high dimensionality. We believe that in higher dimension, SMC-ABC would require a very large number of particles, and a higher number of iteration each of which would cost much more in resampling, leading to a disastrous computational cost.

\begin{figure}
\centering
\includegraphics[scale=.7]{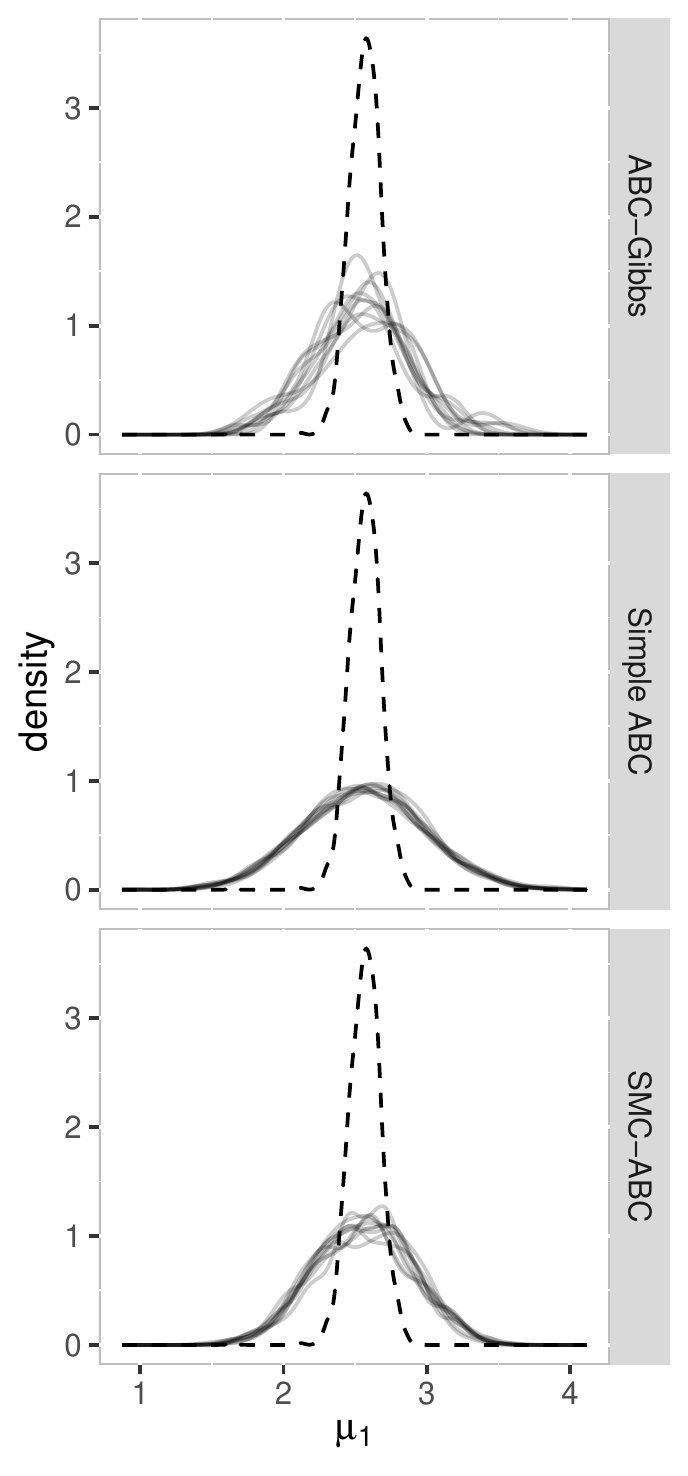} 
\includegraphics[scale=.7]{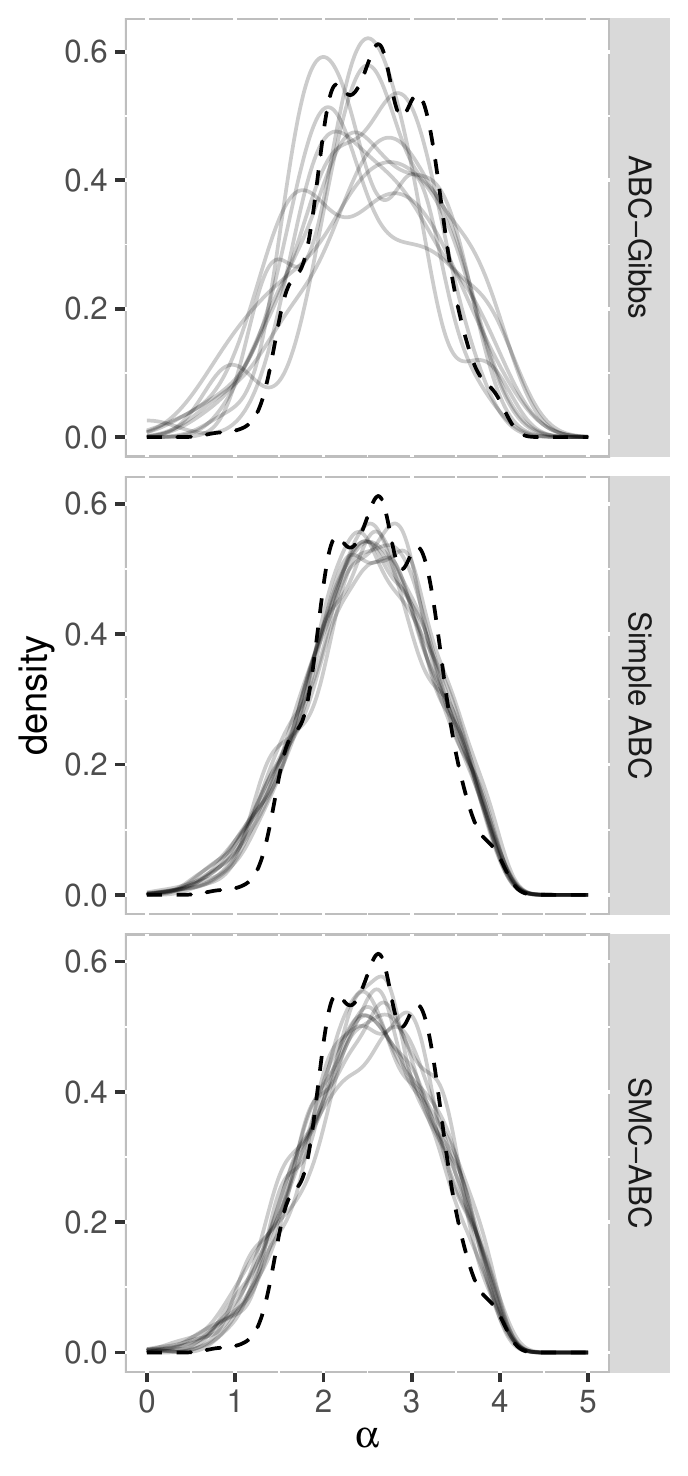} 

\caption{Parameter and hyperparameter for the Normal-Normal model with only two parameter and one hyperparameter.}
\label{fig:toy1D}
\end{figure}

\subsection{Code example}

\begin{lstlisting}[language=R,basicstyle=\ttfamily]


gibbsparam <- function(data, hyper, var, sigm, nbeps, qq) {
  #Gibbs step for the parameter, prior induced by the hyperparameter
  p = length(data)
  thetc = numeric(p)
  dists=0
  for (i in 1:p) { 
    thettest = rnorm(nbeps, hyper, var)
    test = rowMeans(matrix(rnorm(qq * nbeps, thettest, sigm), nrow = nbeps))
    dist = abs(test - data[i])
    thetc[i] = thettest[which.min(dist)]
    dists=dists+min(dist)
  }
  return(list(thetc,dists))
}



gibbshyper <- function(thet, nbeps2, var) {
  #Gibbs step for the hyperparameter, with uniform prior
  res = runif(nbeps2, -4, 4)
  test = rowMeans(matrix(rnorm(length(thet) *
  	nbeps2, res, var), ncol=length(thet)))
  dist = abs(test - mean(thet))
  return(list(res[which.min(dist)],min(dist)))
}

gibbstot <- function(data, thetini, hyperini, sigm, var, 
nbeps1, nbeps2, nbpts, qq) {
  #full function
  reshyper = rep(NA, nbpts + 1)
  resparam = matrix(NA, ncol = nbpts + 1, nrow = length(thetini))
  reshyper[1] = hyperini
  resparam[,1] = thetini
  resdist = rep(NA,nbpts)
  for (i in 2:(nbpts + 1)) {
    resdist[i-1]= 0
    VV=gibbsparam(data, reshyper[i - 1], var, sigm, nbeps1, qq)
    resparam[,i] = VV[[1]]
    resdist[i-1]= resdist[i-1]+VV[[2]]
    WW=gibbshyper(resparam[, i], nbeps2, var)
    reshyper[i] = WW[[1]]
  }
  return(list(reshyper,resparam,resdist))
}

\end{lstlisting}

This exemplar code presents a simple implementation of \ABCG for the
hierarchical normal model. The main function, \texttt{gibbstot} run iterations
of Gibbs, it starts with an initial point \texttt{thetini, hyperini} and sample
successively from the approximate conditionals for \texttt{nbpts} iterations.

Each of the approximate conditionals is sampled from by the functions
\texttt{gibbsparam} and \texttt{gibbshyper}, as their name indicate the first
one samples from $\pi_{\varepsilon_\mu}( \cdot \mid x,\alpha ) $ and the second
one $ \pi_{\varepsilon_\alpha}( \cdot \mid x,\bmu )$. Each one relies on the
use of a fixed sized reference table of size \texttt{nbeps1} et
\texttt{nbeps2}. After having simulated points from the prior, pseudo data is
simulated and compared (in the vectors \texttt{test} and \texttt{dist}) we
return the point with smallest distance.

The other parameters \texttt{var, sigm, qq} correspond to the
variance of the parameter given the hyperparameter, the variance of the
observation given the parameters and the number of observation for each
parameter, respectively.

\section{Implementation of SMC-ABC}\label{sec:SMC-ABC}

Our implementation of SMC-ABC merges the implementations of \cite{delmoral:doucet:jasra:2012} and \cite{toni2008}, in order to avoid degeneracy and arbitrary choice of the thresholds, as described in Algorithm \ref{SMC-ABC}.

\begin{algorithm}[!h]
\KwIn{number of iterations $T$, $M\geq 1$, $\varepsilon_0$, $N\geq 0$, $N_{min}\leq N$, $\alpha>0$.}
\KwOut{a sample $(\theta_1^0,\dots , \theta_N^0)$ from an initial distribution $\pi^0$.}
Compute for each of the value $\theta_i^0$ $M$ pseudo observations and the associated statistic stored in a vector $s^i_0$ of $M$ statistics\;
Set $\varepsilon_0 = \max_i(s_i^0)$\;
Set the weights $w_i^0=1/N$\;
\For{$j=1,\ldots,T$}{
	Compute $\varepsilon_j$ and $w_j^i$ by solving $ESS(\lbrace w^i_j \rbrace, \varepsilon_{j}) = \alpha ESS(\lbrace w^i_{j-1} \rbrace, \varepsilon_{j-1}) $,
	where,
	\[ w_j^i \propto w_{j-1}^i \frac{\sum_{k=1}^M \mathbf{1}_{<\varepsilon_j}(s^i_{j-1}[k])}{\sum_{k=1}^M \mathbf{1}_{<\varepsilon_{j-1}}(s^i_{j-1}[k])} \]
  \If{$ESS(\lbrace w^i_j \rbrace, \varepsilon_{j}) < N_{min}$}{
  Resample the value of the particles according to the weights $w_j^i$ \; We abusively use the same notation for the values after this step \;
  Set the new weights $w_j^i = 1/N$
  }
  \For {$i = 1, \dots , N$}{
  \Do {$\sum_{k=1}^M \mathbf{1}_{<\varepsilon_j}(s^{i*}_{j}[k]) =0$}{Sample $\theta_j^{i *}\sim K_j(\theta_j^i)$ and associated pseudo observations and statistics $s_j^{i*}$.}
  Set $\theta_j^i = \theta_j^{i*}$\;
  }
}
\caption{SMC-ABC}
\label{SMC-ABC}
\end{algorithm}

As underlined in \cite{delmoral:doucet:jasra:2012}, we can choose the kernel $K_j$ so that it depends on the value of the particles. Following custom, we choose $K_j$ to be a Gaussian kernel with covariance matrix $2*Corr(\theta_j^1 , \dots , \theta_j^N)$.

\section{Supplementary material: Moving average example}

\label{sec:MA}

\subsection{Model and implementation}
\label{sub:MA2}

\begin{figure}
\center
\begin{tikzpicture}

\node (a) at (-2.4, 1.2) {$\alpha$};
\node (b1) at (-1.2, 2.4) {$\mu_{1}$};
\node (b2) at (-1.2, 1.5)  {$\mu_{2}$};
\node at (-1.2, 0.75) {$\vdots$};
\node (b3) at (-1.2, 0) {$\mu_{n}$};

\draw [->] (a)--(b1);
\draw [->] (a)--(b2);
\draw [->] (a)--(b3);

\node (x1) at (0, 2.4) {$x_{1}$};
\node (x2) at (0, 1.5) {$x_{2}$};
\node (x3) at (0, 0) {$x_{n}$};
\node at (0, 0.75) {$\vdots$};

\draw [->] (b1)--(x1);
\draw [->] (b2)--(x2);
\draw [->] (b3)--(x3);

\node (s) at (2.4, 1.2) {$\varsigma$};
\node (s1) at (1.2, 2.4) {$\sigma_{1}$};
\node (s2) at (1.2, 1.5) {$\sigma_{2}$};
\node (s3) at (1.2, 0) {$\sigma_{n}$};
\node at (1.2, 0.75) {$\vdots$};

\draw [->] (s)--(s1);
\draw [->] (s)--(s2);
\draw [->] (s)--(s3);
\draw [->] (s1)--(x1);
\draw [->] (s2)--(x2);
\draw [->] (s3)--(x3);

\end{tikzpicture}

\caption{Hierarchical dependence structure used in the application of Section \ref{sec:stars}.}
\label{models}
\end{figure}
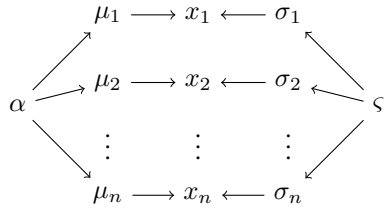

In this section, we study a hierarchical moving average model. A graphical representation of the hierarchy is shown in Figure \ref{models}. We denote $\mathcal{MA}_2(\mu,\sigma^2)$ the distribution of a second order moving average model with parameters $\mu=(\mu_1,\mu_2)$ and $\sigma^2$, that is:
\[
x(t) = y_t + \mu_1y_{t-1}+\mu_2y_{t-2}\,, \qquad \text{with }y_t \sim \mathcal{N}(0,\sigma^2) \text{ for integer }t\geq -1.
\]

We consider a hierarchical version of the $\mathcal{MA}_2$ model, consisting of $n$ parallel observed series and $3n+5$ parameters with the following dependencies and prior distributions: for $j = 1,\ldots, n$,
\[
x_j \sim \mathcal{MA}_2(\mu_j,\sigma^2_j)\,,\quad \sigma_j^2 \sim \mathcal{IG}(\varsigma_1,\varsigma_2)\,,\quad\mu_j = \left(\beta_{j,1} - \beta_{j,2},2(\beta_{j,1}+\beta_{j,2})-1\right) = (\mu_{j,1},\mu_{j,2})\,,
\]
where $(\beta_{j,1},\beta_{j,2},1-\beta_{j,1}-\beta_{j,2}) \sim \mathcal{D}ir(\alpha_1,\alpha_2,\alpha_3)$, and, if $\mathcal E$ denotes the exponential distribution and $\mathcal C_+$ the standard half-Cauchy distribution,
\[
\alpha = (\alpha_1,\alpha_2,\alpha_3) \sim\mathcal{E}(1)^{\otimes 3}\,,\quad
\varsigma  = (\varsigma_1,\varsigma_2) \sim\mathcal{C}_+^{\otimes 2}.
\]

We denote $w(x_j)$ the distance between the first two autocorrelations of $x_j$ and $x_j^{\star }$:
\[
w^2(x_j) =  \{\rho_1(x_j) - \rho_1(x^{\star}_j)\}^2 +  \{\rho_2(x_j) - \rho_2(x^{\star}_j)\}^2,
\]
and
\[
\overline{x_j} = \frac{1}{\lfloor T/3 \rfloor} \sum_{t=1}^{\lfloor T/3 \rfloor} x_j(3t)\,,\quad
v(x_j) = \frac{1}{\lfloor T/3 \rfloor} \left\vert \sum_{t=1}^{\lfloor T/3 \rfloor} (x_j(3t) - \overline{x_j} )^2 - \sum_{t=1}^{\lfloor T/3 \rfloor} (x_j^{\star}(3t) - \overline{x_j^{\star}})^2 \right\vert\,,
\]
where $T$ is the length of the time series. The rationale is that for a $\mathcal{MA}_2$ model $x(t)$ and $x(t+3)$ are independent. Vanilla \ABCS uses a related single pseudo-distance defined by 
\begin{equation}\label{eq:singledistance}
\delta(x)= \sum_{j=1}^n \left\{\frac{w(x_j)}{q_j}+\frac{v(x_j)}{q_j'}\right\}\,,
\end{equation}
where $q_j$ and $q_j'$ are the 0.1\% quantiles of $w(x_j)$ and $v(x_j)$, respectively. This choice is constrained by the fact that these quantities appear to have undefined mean and variance.

For the current model, we have the following implementation:
First, the $\mu_{j}$'s are updated using the pseudo-distance $d_{\mu_j}(x_j,x_j^*)=w(x_j)$.

 Second, the update of $\alpha$ relies on the sufficient statistic associated with Dirichlet distributions:
\[
\boldsymbol{\mu} \mapsto \left(\sum_j\log \{(\mu_{j,2} + 2\mu_{j,1} +1)/4\}, \sum_j \log\{(\mu_{j,2} + 2\mu_{j,1} +1)/4 \}-\mu_{j,1})\right)
\]
Third, the $\sigma_j$'s are updated using the pseudo-distance $d_{\sigma_j}(x_j,x_j^*)=v(x_j)$. And last, $\varsigma$ is updated using the standard sufficient statistic associated with gamma distributions.

The two algorithms output samples from the two pseudo-posteriors. To compare
the efficient of both samplers, we simulate new synthetic data from each
parameter set in the output, and compute the distance \eqref{eq:singledistance}
between observed and simulated samples, which is the distance used by \ABCS.
If \ABCG \ produces a smaller value than the \ABCS\ sampler
associated with this distance, this is an indicator of a better fit of the
\ABCG~distribution with the true posterior.  As in the previous experiment,
the total number of simulations of the time series is used 
as the default measure of the computational cost for the associated algorithm.

\subsection{Toy dataset}\label{sec:toy}

Consider a synthetic dataset of $n=5$ times series each with length $T=100$.
Both samplers return samples of size $N=1000$. The hyperparameters used to produce the
true parameters and the simulated observed series are $\alpha=(1,2,3)$ and $\varsigma = (1,1)$. In \ABCG,
the $\mu_{j}$'s are updated based on $N_\mu=1000$ time series, while the other parameters are updated based on $N_\alpha=N_\sigma=N_\varsigma=100$ replicas.
The overall computational cost for  \ABCG is $\Ntot=5.5 \cdot 10^6$, also used by \ABCS to run
$1.1 \cdot 10^6$ simulations of the whole hierarchy. The computational cost is
slightly superior for \ABCS, as we have to simulate many more Dirichlet and Gamma random variables.

When evaluating the mean of the posterior predictive distance
\eqref{eq:singledistance}, \ABCG achieves an average of $274.1 \pm 2.5$, and
\ABCS an average of $436.8 \pm 1.6$, based on 100 replicates. The sample output by
\ABCG thus offers a noticeably better quality than the one generated by \ABCS from this perspective.
The \ABCS output barely differs from a simulation from the prior, as shown in Figure \ref{toyMoving average}
 for the parameter $\mu_1$. 


\begin{figure}
\center
\includegraphics[scale=.45]{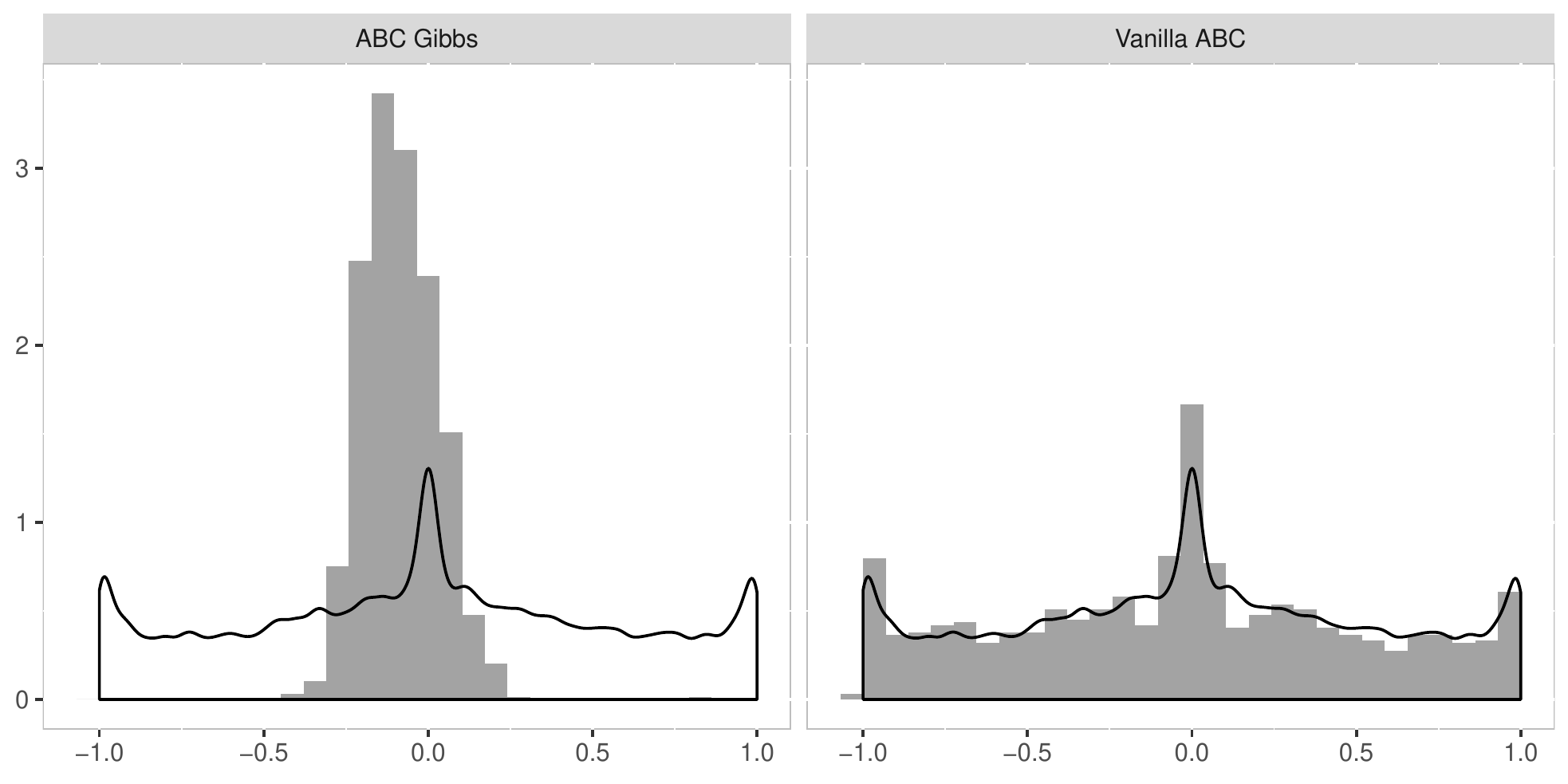} 
\caption{For the toy dataset of subsection \ref{sec:toy}, approximate posterior of $\mu_1$ compared with the prior for \ABCG\ (left) and \ABCS\ (right). The true value was $-0.06$.}
\label{toyMoving average}
\end{figure}

\subsection{Stellar flux}\label{sec:stars}

We now apply this model to stellar flux data. The 8GHz daily flux emitted by seven stellar objects is analysed in
\cite{Laziodualfrequency}, and the data were made public by the Naval Research Laboratory: \url{https://tinyurl.com/yxorvl4u}.
Once a few missing observations have been removed, \cite{Laziodualfrequency} suggest that the model 
described in Section \ref{sub:MA2} may be well suited to these data, with $T=208$.
In \ABCG, the $\mu_{j}$'s are updated based on $N_\mu=500$ time series,
while the other parameters require $N_\alpha = N_\varsigma = N_\sigma=100$ replicas. (The overall computing time is the same for the toy and the current datasets, that is, one hour on an Intel Xeon CPU E5-2630 v4 with rate 2.20GHz.)

The average posterior distance to the observed sample is  $232.8\pm1.25$ for
\ABCG and $535\pm0.95$ for \ABCS. The poor fit of the latter is confirmed in
Figure \ref{RealMoving average}, as it again stays quite close to the prior for the $\mu$'s. Since our
model differs from the one proposed in \cite{Laziodualfrequency}, estimators cannot be directly compared.

\begin{figure}
\center
\includegraphics[scale=.45]{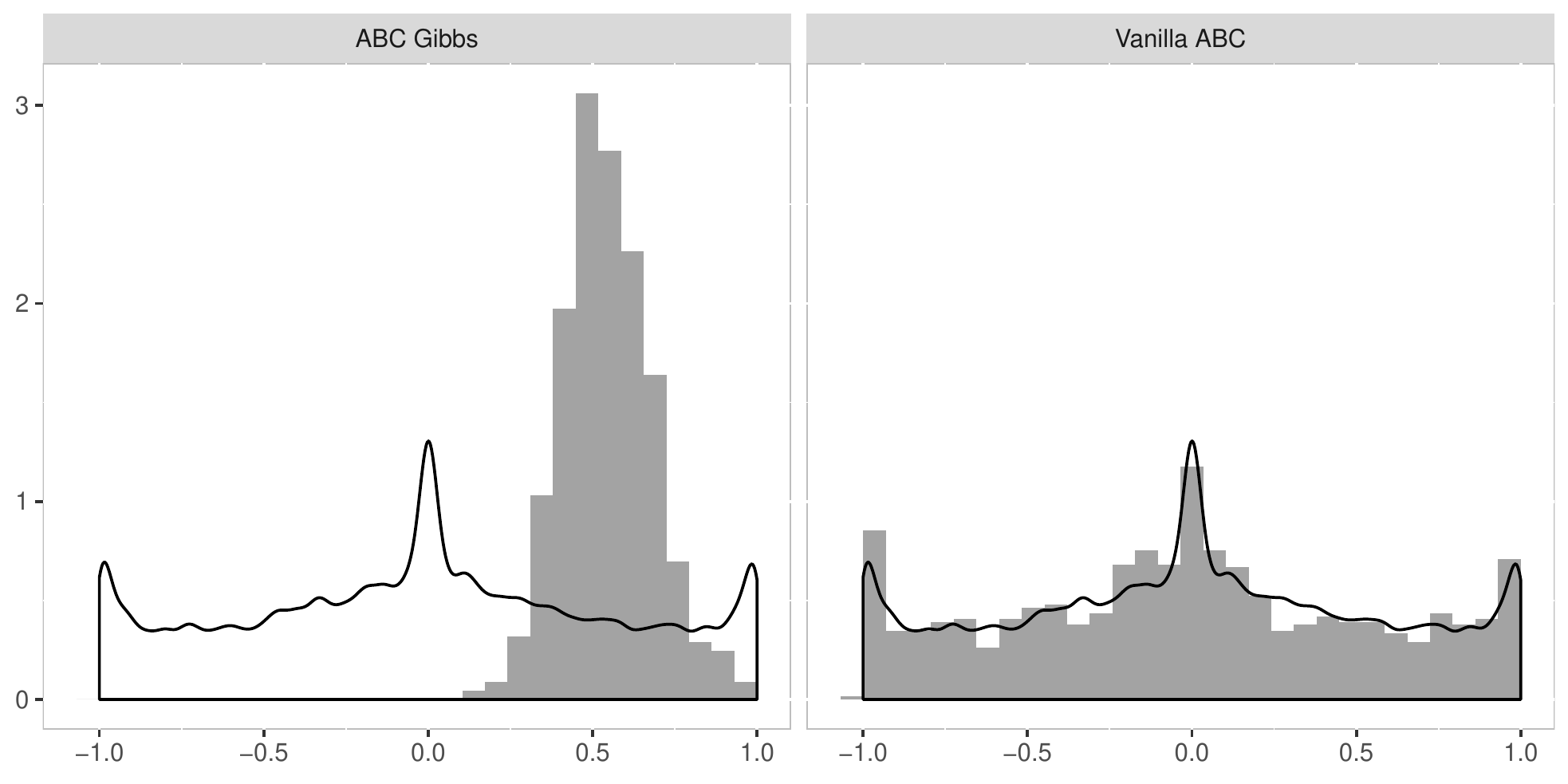} 
\caption{\label{RealMoving average} For the stellar dataset of subsection \ref{sec:stars}, approximate posterior of $\mu_1$ compared with the prior for \ABCG\ (left) and \ABCS\ (right)}
\end{figure}

\section{Supplementary material: proofs of theorems}

In this supplementary material, we define $\Theta_j$ as the domain of $\theta_j$. For the proofs that pertain to model \eqref{eq:modelun} we define $\mathcal{A}$ as the domain of $\alpha$  and $\mathcal{B}$ as the domain of $\mu$. For a space $E$, $\mathcal{P}(E)$ is the space of the probability distributions over $E$.

\subsection{Generalities on total variation distance}

The main tool in our proofs is the total variation distance used by \citet{nummelin:1978} and \citet{meyn:twee:1993}. Let $\nu$ and $\tilde{\nu}$ be two probability distributions over the same space $E$. A coupling $\gamma$ between $\nu$ and $\tilde{\nu}$  is a probability distribution on $E \times E$ such that $\int \gamma(x,y)\mathrm{d}x = \nu$ and $\int \gamma(x,y)\mathrm{d} y = \tilde{\nu}$. Let $\Gamma(\nu,\tilde{\nu})$ denote the set of all couplings between $\nu$ and $\tilde{\nu}$. Then the total variation distance is defined as
\[
\Vert \nu - \tilde{\nu} \Vert_{TV} =\frac{1}{2} \inf_{\gamma \in \Gamma(\nu,\tilde{\nu})}  \mathrm{pr}(x \neq y \mid (x,y) \sim \gamma).
\]
 To handle this distance, we build an explicit coupling between the distributions: this provides an upper bound on the total variation distance. Note that there always exists an optimal coupling between two distributions, that is a coupling $\gamma_0$ such that $\Vert \nu - \tilde{\nu} \Vert_{TV} = \frac{1}{2} \mathrm{pr}(x \neq y \mid (x,y) \sim \gamma_0) $.

\subsection{Proof of Theorem \ref{THgen}}\label{proof:THgen}

In this proof, we drop the conditionings on $x^\star$, $s_1$, and $s_2$, as they have no use in the computations and create a notational burden.

We only need to prove that the Markov chain $(\theta_1^{(i)})_{1\leq i\leq N}$ has a stationary distribution. We show that $Q \ : \mathcal{P}(\Theta_1) \rightarrow \mathcal{P}(\Theta_1)$, the mapping associated with the transition kernel, is a contraction; that is, we prove that there exists $L>1$ such that for all $\nu$ and $\tilde\nu$ in $\mathcal P(E)$
\[
\Vert Q\nu - Q\tilde{\nu} \Vert_{TV} \leq L \Vert \nu - \tilde{\nu} \Vert_{TV}.
\]
To build a coupling between $Q\nu$ and $Q\tilde{\nu}$ we construct a coupling kernel $\tilde{Q} : \mathcal{P}(\Theta_1 \times \Theta_1) \rightarrow \mathcal{P}(\Theta_1 \times \Theta_1) $, which takes a coupling $\xi_0$ as argument, such that $\int \tilde{Q}\xi_0(x,y) \mathrm{d}x = Q\nu(y) $ and $\int \tilde{Q}\xi_0(x,y) \mathrm{d}y = Q\tilde{\nu}(x)$. This coupling kernel is explicitly defined by the following procedure, which takes as input $(\theta_1,\tilde{\theta}_1)\sim \xi_0$ a coupling of $\nu$ and $\tilde{\nu}$, and returns $(\theta_1',\tilde{\theta}_1') \sim \tilde{Q}\xi_0$:
\begin{algorithm}[h!]
\KwIn{$(\theta_1,\tilde{\theta}_1) \sim \xi_0$, $\xi_1( \cdot \mid \theta_1,\tilde{\theta}_1)$ an optimal coupling between $\pi_{\varepsilon_2}(\cdot \mid \theta_1)$ and $\pi_{\varepsilon_2}(\cdot \mid \tilde{\theta}_1)$,  $\xi_2( \cdot \mid \theta_2,\tilde{\theta}_2)$ an optimal coupling between $\pi_{\epsilon_1}(\cdot \mid \theta_2)$ and $\pi_{\varepsilon_1}(\cdot \mid \tilde{\theta}_2)$.}
\KwOut{$(\theta_1',\tilde{\theta}_1') \sim \tilde{Q}\xi_0$.}
$(\theta_2,\tilde{\theta}_2) \sim \xi_1(\cdot \mid \theta_1,\tilde{\theta}_1)$;

$(\theta_1',\tilde{\theta}_1') \sim \xi_2(\cdot \mid \theta_2, \tilde{\theta}_2) $.
\caption{Coupling procedure for Theorem \ref{THgen}}
\end{algorithm}

This procedure satisfies the property that if $\theta_1=\tilde{\theta}_1$ then $\theta_1'=\tilde{\theta}_1'$, since for any distribution $\nu_0$, the optimal coupling between $\nu_0$ and itself is $(x,y) \mapsto \nu_0(x)\delta_{x=y} $.

The proofs choose $\xi_0$ as the optimal coupling between $\nu$ and $\tilde{\nu}$. In the following, $\tilde{\gamma} =\tilde{Q}\xi_0$, so that
\begin{align*}
\begin{split}
\Vert Q\nu - Q\tilde{\nu} \Vert_{TV}  &= \frac{1}{2} \inf_{\gamma \in \Gamma(Q\nu,Q\tilde{\nu})}\mathrm{pr}(\theta_1' \neq \tilde{\theta}_1' \mid (\theta_1,\tilde{\theta}_1) \sim \gamma) \\
&\leq \frac{1}{2} \mathrm{pr}(\theta_1' \neq \tilde{\theta}_1' \mid \theta_1 = \tilde{\theta}_1, (\theta_1',\tilde{\theta}_1') \sim \tilde{\gamma})\mathrm{pr}_{\xi_0}(\theta_1 = \tilde{\theta}_1) \\
& \hspace{1cm}+  \frac{1}{2} \mathrm{pr}(\theta_1' \neq \tilde{\theta}_1' \mid \theta_1 \neq \tilde{\theta}_1, (\theta_1',\tilde{\theta}_1') \sim \tilde{\gamma})\mathrm{pr}_{\xi_0}(\theta_1\neq \tilde{\theta}_1) \\
&\leq \frac{1}{2} \mathrm{pr}(\theta_1' \neq \tilde{\theta}_1' \mid \theta_1 \neq \tilde{\theta}_1, (\theta_1',\tilde{\theta}_1') \sim \tilde{\gamma})\mathrm{pr}_{\xi_0}(\theta_1 \neq \tilde{\theta}_1)\\
& \leq \Vert \nu - \tilde{\nu} \Vert_{TV}\mathrm{pr}(\theta_1' \neq \tilde{\theta}_1' \mid \theta_1 \neq \tilde{\theta}_1, (\theta_1' \neq\tilde{\theta}_1') \sim \tilde{\gamma}).\\
\end{split}
\end{align*}
It is now sufficient to bound $\mathrm{pr}(\theta_1' \neq
\tilde{\theta}_1' \mid \theta_1 \neq \tilde{\theta}_1,
(\theta_1',\tilde{\theta}_1') \sim \tilde{\gamma}) = 1 -
\mathrm{pr}(\theta_1' = \tilde{\theta}_1' \mid \theta_1 \neq
\tilde{\theta}_1, (\theta_1',\tilde{\theta}_1') \sim \tilde{\gamma}) $, that
is to find a lower bound on the probability that two different values $\theta_1$
and $\tilde{\theta}_1$ transition to the same value.

If $\theta_2=\tilde{\theta}_2$ then necessarily, $\theta_1'=\tilde{\theta}_1'$,
in other words, if the coupling is successful at the first step of the
procedure it is sufficient. This means that a lower bound on the coupling
probability is the coupling probability at the first step of the procedure.
Now,

\begin{align*}
\mathrm{pr}(\theta_1' = \tilde{\theta}_1' \mid \theta_1 \neq \tilde{\theta}_1, (\theta_1',\tilde{\theta}_1') \sim \tilde{\gamma}) \geq& 1-2 \Vert \pi_\varepsilon( \cdot \mid \theta_1) - \pi_\varepsilon( \cdot \mid \tilde{\theta}_1) \Vert_{TV}\\
\geq& 1-2\kappa >0.
\end{align*}

This proves that the map $ Q: \nu \mapsto Q\nu $ is a contraction. The space of
all measures on $\mathcal{A}$ is complete when endowed with the total variation
distance. Furthermore, the subspace of all probability distributions on
$\Theta_1$ is stable by $Q$. Hence, by the Banach fixed-point theorem,
it enjoys a fixed point and in particular the sequence $ (Q^n\pi) $, with $\pi$
an arbitrary prior distribution, converges to this fixed point with rate
$1-2\kappa$.

\subsection{Proof of Theorem \ref{THgenconv}}

The assumptions on $L_2$ and $L_0$ imply with the triangular inequality that the assumptions of Theorem \ref{THgen} are verified, and thus that $\mu_\varepsilon$ exists.

In this proof, we need a coupling between two chains with different transition
kernels. Let $\nu_\varepsilon$ be the target distribution of the approximate
Gibbs sampler and $\nu_0$ be the target distribution of the exact Gibbs
sampler. Let $(\theta_1,\tilde{\theta}_1)$ be a realisation of an optimal
coupling $\xi_0$ between $\nu_\varepsilon$ and $\nu_0$. As before we propose a coupling procedure:

\begin{algorithm}[h!]
\KwIn{$(\theta_1,\tilde{\theta}_1) \sim \xi_0$, $\xi_3(\cdot \mid \theta_1,\tilde{\theta}_1)$ an
optimal coupling between $ \pi_\varepsilon( \cdot \mid \theta_1) $  and $\pi(
\cdot \mid \tilde{\theta}_1)$, $\xi_4(\cdot \mid
\theta_2,\tilde{\theta}_2)$ an optimal coupling between $ \pi_\eta( \cdot \mid
\theta_2) $  and $\pi( \cdot \mid \tilde{\theta}_2)$.}
\KwOut{$(\theta_1',\tilde{\theta}_1') \sim \tilde{Q}\xi_0$.}
$(\theta_2,\tilde{\theta}_2) \sim \xi_3(\cdot \mid \theta_1,\tilde{\theta}_1)$;

$(\theta_1',\tilde{\theta}_1') \sim \xi_4(\cdot \mid \theta_2, \tilde{\theta}_2) $.
\caption{Coupling procedure for Theorem \ref{THgenconv}}
\end{algorithm}

As the distributions $\nu_\varepsilon$ and $\nu_0$ are stationary for the evolution process, we have
\begin{align*}
\mathrm{pr}(\theta_1' \neq \tilde{\theta}_1') =& \mathrm{pr}(\theta_1' \neq \tilde{\theta}_1' \mid \theta_1 \neq \tilde{\theta}_1 )\mathrm{pr}(\theta_1 \neq \tilde{\theta}_1) +\mathrm{pr}(\theta_1' \neq \tilde{\theta}_1'\mid \theta_1 = \tilde{\theta}_1)\mathrm{pr}(\theta_1' = \tilde{\theta}_1') \\
\leq& \frac{\mathrm{pr}(\theta_1' \neq \tilde{\theta}_1'\mid \theta_1 = \tilde{\theta}_1)}{\mathrm{pr}(\theta_1' = \tilde{\theta}_1'\mid \theta_1 \neq \tilde{\theta}_1)}.
\end{align*}
As before we use a rough bound on the denominator:
\begin{align*}
\mathrm{pr}(\theta_1' = \tilde{\theta}_1'\mid \theta_1 \neq \tilde{\theta}_1) \geq& (1 - 2 \sup_\varepsilon\sup_{\theta_1,\tilde{\theta}_1} \Vert \pi_\varepsilon( \cdot \mid \theta_1) - \pi( \cdot \mid \tilde{\theta}_1)\Vert_{TV}) \\
\geq& 1-2L_0.
\end{align*}
For the numerator, we have, with $\theta_2$ and $\tilde{\theta}_2$ the transitory values of the second parameter, 
\begin{align*}
\mathrm{pr}(\theta_1' \neq \tilde{\theta}_1'\mid \theta_1 = \tilde{\theta}_1) &\leq \mathrm{pr}(\theta_1' \neq \tilde{\theta}_1' \mid \theta_2 = \tilde{\theta}_2) \mathrm{pr}(\theta_2 \neq \tilde{\theta}_2 \mid \theta_1 = \tilde{\theta}_1) \\
& \hspace{2cm}+ \mathrm{pr}(\theta_1' \neq \tilde{\theta}_1' \mid \theta_2 \neq \tilde{\theta}_2) \mathrm{pr}(\theta_2 \neq \tilde{\theta}_2 \mid \theta_1 = \tilde{\theta}_1)\\
&\leq \sup_{\theta_2}\mathrm{pr}\{\theta_1 \neq \tilde{\theta}_1 \mid (\theta_1,\tilde{\theta}_1) \sim \xi_4(\cdot \mid \theta_2,\theta_2)\} + \sup_{\vartheta_1} \mathrm{pr}\{\theta_2 \neq \tilde{\theta}_2 \mid (\theta_2,\tilde{\theta}_2)\sim \xi_3(\cdot \mid \vartheta_1,\vartheta_1)\} \\
&\leq 2L_1(\varepsilon_1) + 2L_2(\varepsilon_2) 
\end{align*} 
Putting together both estimates gives the bound of the theorem.

\subsection{Proofs specific to the hierarchical case}
\label{specificproofs}

In addition to the general theorems presented in the main paper, we provide in this subsection 
convergence results which are specific to hierarchical models, with assumptions which may be more intuitive or 
easier to verify.
specific convergence results. 
They are based on a particular implementation of
ABC-Gibbs, presented for $n=1$ and in the case of an analytically available
conditional density $\pi(\mu \mid \alpha,x^\star)$, in Algorithm \ref{algohierar}. We will
gradually weaken the assumptions to finally prove Theorem \ref{thm:k1k2k3}:

\begin{theorem}
Assume there exists a non-empty convex set $C$ with positive prior measure such that
\begin{align*}
\kappa_1=&\inf_{s_\alpha(\mu) \in C} \pi(B_{s_\alpha(\mu),\epsilon_\alpha/4})>0\,, \\
\kappa_2=&\inf_{\alpha} \inf_{s_\alpha(\mu) \in C} \pi_{\varepsilon_\mu}\{B_{s_\alpha(\mu),3\epsilon_\alpha/2} \mid s_\mu(x^\star,\alpha),\alpha \}>0\,, \\
 \kappa_3 =& \inf_{\alpha} \pi_{\varepsilon_\mu}\{s_\alpha(\mu) \in C \mid s_\mu(x^\star,\alpha)\} > 0\,,
\end{align*}
where $B_{z,h}$ denotes the ball of centre $z$ and radius $h$. Then the Markov chain produced by Algorithm \ref{algohier} converges geometrically in
total variation distance to a stationary distribution $\nu_{\varepsilon}$, with
geometric rate  $1-\kappa_1\kappa_2\kappa_3^2 $.
\label{thm:k1k2k3}
\end{theorem}

The rate in Theorem \ref{thm:k1k2k3} is uninformative, as it is specific to the selected implementation. 

\begin{algorithm}[h!]
\KwIn{$\alpha^{(0)} \sim \pi(\alpha)$, $\mu^{(0)} \sim \pi(\mu \mid \alpha^{(0)})$.}
\KwOut{A sample $(\alpha^{(i)},\mu^{(i)})_{1 \leq i\leq N}$.}
\For{$i = 1 , \dots, N$}{
 $\mu^c \sim \pi(\cdot \mid \alpha^{(i-1)}, x^\star)$
 
 $\alpha^c \sim  \pi$
 
 $\tilde{\mu} \sim \pi(\cdot \mid \alpha^c)$
 
 \eIf{$\eta\{s_\alpha({\mu}),s_\alpha(\mu^c)\} <  \epsilon_\alpha$}{
 $\mu^{(i+1)} \leftarrow \mu^c$
 
 $\alpha^{(i+1)} \leftarrow \alpha^c$
}{
 $\mu^{(i+1)} \leftarrow \mu^{(i)}$
 
 $\alpha^{(i+1)} \leftarrow \alpha^{(i)}$
}
}
\caption{Implementation of ABC-Gibbs used in the proofs.}
\label{algohierar}
\end{algorithm}

First we state the most restrictive result:
\begin{theorem}
Assume that the following conditions are both satisfied:
\begin{align*}
\kappa_1&=\inf_{\mu} \pi(B_{s_\alpha(\mu),\epsilon_\alpha/4})>0 \\
\kappa_2&=\inf_{\alpha} \inf_{\mu} \pi(B_{s_\alpha(\mu),3\epsilon_\alpha/2} \mid \alpha,x)>0.
\end{align*}
Then, the Markov chain associated with Algorithm \ref{algohierar} enjoys an
invariant distribution, and it converges geometrically to this invariant measure
with rate $1-\kappa_1\kappa_2$ for the total variation distance.
\end{theorem}

\begin{proof}

The technique of the proof is essentially similar to that of Theorem \ref{THgen}. Let
$\nu$ and $\tilde{\nu}$ be two distributions over $\mathcal{A}$. We describe the
evolution of $\alpha,\tilde{\alpha}$ into $\alpha',\tilde{\alpha}'$, though the kernel $\tilde{Q}$. We denote
$\mu,\tilde{\mu}$ the transitory second parameter.

\begin{algorithm*}[h!]
\KwIn{$(\alpha,\tilde{\alpha})$.}
\KwOut{$(\alpha',\tilde{\alpha}')$.}
\eIf{$\alpha \neq \tilde{\alpha}$}{$(\mu,\tilde{\mu})\sim\pi(\cdot \mid \alpha,x)\otimes \pi(\cdot \mid \tilde{\alpha},x)$}{$\mu=\tilde{\mu} \sim \pi(\mu \mid \alpha, x)$}

$\alpha^c \sim \pi$

$\mu^c \sim \pi(\cdot \mid \alpha^c)$;

\eIf{$\eta\{s_\alpha(\mu),s_\alpha(\mu^c)\} \leq \varepsilon_\alpha $}{$\alpha' \leftarrow \alpha^c$ }{$\alpha'\leftarrow\alpha$}
\eIf{$\eta\{s_\alpha(\tilde{\mu}),s_\alpha(\mu^c)\} \leq \varepsilon_\alpha $}{$\tilde{\alpha}' \leftarrow \alpha^c$}{ $\tilde{\alpha}'\leftarrow\tilde{\alpha}$}
\caption{Coupling procedure}
\end{algorithm*}

This process defines a transition kernel $\tilde{Q}$ for two coupled chains. As in the previous proofs, if $\alpha=\tilde{\alpha}$ then $\alpha'=\tilde{\alpha}'$.

Let  $(\alpha,\tilde{\alpha}) \sim \xi$, an optimal coupling between $\nu$ and $\tilde{\nu}$. Then,
\begin{align*}
\begin{split}
\Vert Q\nu - Q\tilde{\nu} \Vert_{TV}  &= \frac{1}{2} \inf_{\gamma \in \Gamma(Q\nu,Q\tilde{\nu})}\mathrm{pr}(\alpha' \neq \tilde{\alpha}' \mid (\alpha',\tilde{\alpha}') \sim \gamma) \\
&\leq \frac{1}{2} \mathrm{pr}_{\xi}(\alpha' \neq \tilde{\alpha}' \mid \alpha = \tilde{\alpha}, (\alpha',\tilde{\alpha}') \sim \tilde{Q}\xi)\mathrm{pr}_{\xi}(\alpha = \tilde{\alpha}) \\
& \hspace{1cm}+ \frac{1}{2} \mathrm{pr}_{\xi}(\alpha' \neq \tilde{\alpha}' \mid \alpha \neq \tilde{\alpha}, (\alpha',\tilde{\alpha}') \sim \tilde{Q}\xi)\mathrm{pr}_{\xi}(\alpha\neq \tilde{\alpha}) \\
&\leq \frac{1}{2}\mathrm{pr}_{\xi}(\alpha' \neq \tilde{\alpha}' \mid \alpha \neq \tilde{\alpha}, (\alpha',\tilde{\alpha}') \sim \tilde{Q}\xi )\mathrm{pr}_{\xi}(\alpha\neq \tilde{\alpha})\\
& \leq \Vert \nu - \tilde{\nu} \Vert_{TV}\mathrm{pr}_{\xi} (\alpha' \neq \tilde{\alpha}' \mid \alpha \neq \tilde{\alpha}, (\alpha',\tilde{\alpha}') \sim \tilde{Q}\xi ).\\
\end{split}
\end{align*}
It is sufficient to find a uniform upper bound on $\mathrm{pr}_{\nu,\tilde{\nu}}(\alpha' \neq \tilde{\alpha}' \mid \alpha \neq \tilde{\alpha},(\alpha',\tilde{\alpha}') \sim \tilde{Q}\xi) = \int \pi\{s_\alpha(\mu^c) \notin B_{s_\alpha(\mu),\varepsilon_\alpha} \cap B_{s_\alpha(\tilde{\mu}),\varepsilon_\alpha} \} \pi(\mu \mid \alpha,x) \pi(\tilde{\mu} \mid \tilde{\alpha},x)\nu(\alpha)\tilde{\nu}(\tilde{\alpha})\,\mathrm{d} \tilde{\alpha}\mathrm{d}\mu \mathrm{d}\tilde{\mu}$. Notice that we can choose our coupling $\xi$ such that conditionally on $\alpha \neq \tilde{\alpha}$ the marginals are independent.
\begin{align*}
\mathrm{pr}_{\nu,\tilde{\nu}}(\alpha' \neq \tilde{\alpha}' \mid \alpha \neq \tilde{\alpha})=& \int \pi\left\lbrace(B_{s_\alpha(\mu),\varepsilon_\alpha} \cap B_{s_\alpha(\tilde{\mu}),\varepsilon_\alpha})^c\right\rbrace \pi(\mu \mid \alpha,x) \\
&\hspace{1cm} \times \pi(\tilde{\mu} \mid \tilde{\alpha},x)\nu(\alpha)\tilde{\nu}(\tilde{\alpha})\,\mathrm{d}\alpha \mathrm{d} \tilde{\alpha}\mathrm{d}\mu \mathrm{d}\tilde{\mu}\\
=& \int \pi\left\lbrace(B_{s_\alpha(\mu),\varepsilon_\alpha} \cap B_{s_\alpha(\tilde{\mu}),\varepsilon_\alpha})^c\right\rbrace \pi(\mu \mid \alpha,x) \\
&\hspace{1cm}\times \pi(\tilde{\mu} \mid \tilde{\alpha},x)\nu(\alpha)\tilde{\nu}(\tilde{\alpha})\mathbf{1}_{\{\eta\{s_\alpha(\mu),s_\alpha(\tilde{\mu})\}\leq 3\varepsilon_\alpha/2\} }\,\mathrm{d}\alpha \mathrm{d} \tilde{\alpha}\mathrm{d}\mu \mathrm{d}\tilde{\mu}\\
&+ \int \pi\left\lbrace(B_{s_\alpha(\mu),\varepsilon_\alpha} \cap B_{s_\alpha(\tilde{\mu}),\varepsilon_\alpha})^c\right\rbrace \pi(\mu \mid \alpha,x) \\
&\hspace{1cm}\times \pi(\tilde{\mu} \mid \tilde{\alpha},x)\nu(\alpha)\tilde{\nu}(\tilde{\alpha})\mathbf{1}_{\{\eta\{s_\alpha(\mu),s_\alpha(\tilde{\mu})\}> 3\varepsilon_\alpha/2\} }\,\mathrm{d}\alpha \mathrm{d} \tilde{\alpha}\mathrm{d}\mu \mathrm{d}\tilde{\mu} \\
=& I_1 + I_2.
\end{align*}
We now bound $I_1$ and $I_2$.
\begin{align*}
I_1 =&\int \pi\left\lbrace(B_{s_\alpha(\mu),\varepsilon_\alpha} \cap B_{s_\alpha(\tilde{\mu}),\varepsilon_\alpha})^c\right\rbrace \pi(\mu \mid \alpha,x)\pi(\tilde{\mu} \mid \tilde{\alpha},x)\nu(\alpha)\tilde{\nu}(\tilde{\alpha})\mathbf{1}_{\{\eta\{s_\alpha(\mu),s_\alpha(\tilde{\mu})\}\leq 3\varepsilon_\alpha/2\}}\mathrm{d} \tilde{\alpha} \mathrm{d}\alpha \mathrm{d}\mu \mathrm{d}\tilde{\mu} \\
\leq& \int \pi(B_{\frac{s_\alpha(\mu)+s_\alpha(\tilde{\mu})}{2},\varepsilon_\alpha/4}^c) \pi(\mu \mid \alpha,x) \pi(\tilde{\mu} \mid \tilde{\alpha},x)\nu(\alpha)\mathbf{1}_{\{\eta\{s_\alpha(\mu),s_\alpha(\tilde{\mu})\}\leq 3\varepsilon_\alpha/2\}}\nu(\tilde{\alpha})\mathrm{d}\alpha \mathrm{d} \tilde{\alpha}\mathrm{d}\mu \mathrm{d}\tilde{\mu}\\
\leq& \mathrm{pr}_{\nu,\tilde{\nu}}\{\eta\{s_\alpha(\mu),s_\alpha(\tilde{\mu})\}\leq 3\varepsilon_\alpha/2\}\nonumber\\
&\hspace{1cm} - \int \pi(B_{\frac{s_\alpha(\mu) +s_\alpha(\tilde{\mu})}{2},\varepsilon_\alpha/2}) \pi\{s_\alpha(\tilde{\mu}) \in B_{s_\alpha(\mu),3\varepsilon_\alpha/2} \mid \tilde{\alpha},x \}\pi(\mu \mid \alpha,x)\nu(\alpha)\tilde{\nu}(\tilde{\alpha})\mathrm{d}\alpha \mathrm{d} \tilde{\alpha}\mathrm{d}\mu \mathrm{d}\tilde{\mu} \\
\leq& \mathrm{pr}_{\nu,\tilde{\nu}}\{\eta\{s_\alpha(\mu),s_\alpha(\tilde{\mu})\}\leq 3\eta/2\} \\
& \hspace{1cm} -\kappa_1 \int \pi\{s_\alpha(\tilde{\mu}) \in B_{s_\alpha(\mu),3\varepsilon_\alpha/2} \mid \tilde{\alpha},x \}\pi(\mu \mid \alpha,x)\nu(\alpha)\tilde{\nu}(\tilde{\alpha})\mathrm{d}\alpha \mathrm{d} \tilde{\alpha}\mathrm{d}\mu\\
\leq& \mathrm{pr}_{\nu,\tilde{\nu}}\{\eta\{s_\alpha(\mu),s_\alpha(\tilde{\mu})\}\leq 3\varepsilon_\alpha/2\} - \kappa_1\kappa_2 \\
& \nonumber\\
I_2 =& \int \pi\left\lbrace(B_{s_\alpha(\mu),\varepsilon_\alpha} \cap B_{s_\alpha(\tilde{\mu}),\varepsilon_\alpha})^c\right\rbrace \pi(\mu \mid \alpha,x)\pi(\tilde{\mu} \mid \tilde{\alpha},x)\nu(\alpha)\tilde{\nu}(\tilde{\alpha})\mathbf{1}_{\{ \eta\{s_\alpha(\mu),s_\alpha(\tilde{\mu})\}>3\varepsilon_\alpha/2\} }\mathrm{d}\alpha \mathrm{d} \tilde{\alpha}\mathrm{d}\mu \mathrm{d}\tilde{\mu} \\
\leq& \pi(\mu \mid \alpha,x) \pi(\tilde{\mu} \mid \tilde{\alpha},x)\nu(\alpha)\mathbf{1}_{\eta\{s_\alpha(\mu),s_\alpha(\tilde{\mu})\}>3\varepsilon_\eta/2}\tilde{\nu}(\tilde{\alpha})\mathrm{d}\alpha \mathrm{d} \tilde{\alpha}\mathrm{d}\mu \mathrm{d}\tilde{\mu} \nonumber \\
\leq& \mathrm{pr}_{\nu,\tilde{\nu}}\{\eta\{s_\alpha(\mu),s_\alpha(\tilde{\mu})\}>3\varepsilon_\eta/2\}
\end{align*}
Finally, putting both inequalities together, we have $I_1 + I_2 \leq 1-\kappa_1\kappa_2 $, with $\kappa_1\kappa_2>0$ and
\begin{align*}
\Vert Q\nu - Q\tilde{\nu} \Vert_{TV} &\leq (1-\kappa_1\kappa_2) \Vert \nu - \tilde{\nu} \Vert_{TV}.
\end{align*}
The conclusion is the same as in the  proof of Theorem \ref{THgen}.

\end{proof}

\begin{remark}
In the proof, when we describe the coupling kernel, we generate $\mu$ and $\tilde{\mu}$ independently if $\alpha$ and $\tilde{\alpha}$ are different and as a single $\mu$ if they are equal. This is a particular coupling of the distributions $\pi(\cdot \mid \alpha,x) $ and $\pi( \cdot \mid \tilde{\alpha},x)$. Here, the link between Theorem \ref{THgen} and this one becomes clear, as we make the coupling explicit toward reaching a bound in total variation.\end{remark}
%
%
%
%

We now relax the assumptions. First, we remove the assumption that $\mathcal{A}$ is compact: the resulting theorem is Theorem \ref{thm:k1k2k3}.
\begin{proof}
With the same notations as before, we merely need to find a lower bound :
\[ I_3=\int \mathrm{pr}(\mu^c \in B_{\frac{s_\alpha(\mu)+s_\alpha(\tilde{\mu})}{2},\varepsilon_\alpha/2})\pi(\mu \mid \alpha,x)\pi(\tilde{\mu} \mid \tilde{\alpha},x)\mu(\alpha)\nu(\tilde{\alpha})\mathbf{1}_{\{\eta\{s_\alpha(\mu),s_\alpha(\tilde{\mu})\}>3\varepsilon_\alpha/2\}}\mathrm{d}\alpha \mathrm{d} \tilde{\alpha}\mathrm{d}\mu \mathrm{d}\tilde{\mu}   \]
\begin{align*}
I_3 \geq& \int \mathrm{pr}(B_{\frac{s_\alpha(\mu)+s_\alpha(\tilde{\mu})}{2},\varepsilon_\alpha/2})\pi(\mu \mid \alpha,x)\pi(\tilde{\mu} \mid \tilde{\alpha},x)\nu(\alpha)\tilde{\nu}(\tilde{\alpha})\\
&\hspace{1cm}\times\mathbf{1}_{\{s_\alpha(\mu)+s_\alpha(\tilde{\mu})\}/2 \in C\}}\mathbf{1}_{\{\eta\{s_\alpha(\mu),s_\alpha(\tilde{\mu})\}>3\varepsilon_\alpha/2\}}\mathrm{d}\alpha \mathrm{d} \tilde{\alpha}\mathrm{d}\mu \mathrm{d}\tilde{\mu}\\
&+\int \mathrm{pr}(\mu \in B_{\frac{s_\alpha(\mu)+s_\alpha(\tilde{\mu})}{2},\varepsilon_\alpha/2})\pi(\mu \mid \alpha,x)\pi(\tilde{\mu} \mid \tilde{\alpha},x)\nu(\alpha)\tilde{\nu}(\tilde{\alpha}) \\
&\hspace{1cm}\times\mathbf{1}_{\{s_\alpha(\mu)+s_\alpha(\tilde{\mu})\}/2 \notin C\}}\mathbf{1}_{\{\eta\{s_\alpha(\mu),s_\alpha(\tilde{\mu})\}>3\varepsilon_\alpha/2\}}\mathrm{d}\alpha \mathrm{d} \tilde{\alpha}\mathrm{d}\mu \mathrm{d}\tilde{\mu}\\
\geq& \int \mathrm{pr}(\mu \in B_{\frac{s_\alpha(\mu)+s_\alpha(\tilde{\mu})}{2},\varepsilon_\alpha/2})\pi(\mu \mid \alpha,x)\pi(\tilde{\mu} \mid \tilde{\alpha},x)\nu(\alpha)\tilde{\nu}(\tilde{\alpha}) \\
&\hspace{1cm}\times\mathbf{1}_{\{s_\alpha(\mu)\in C\}}\mathbf{1}_{\{s_\alpha(\tilde{\mu}) \in C\}}\mathbf{1}_{\{\eta\{s_\alpha(\mu),s_\alpha(\tilde{\mu})\{>3\varepsilon_\alpha/2\}}\mathrm{d}\alpha \mathrm{d} \tilde{\alpha}\mathrm{d}\mu \mathrm{d}\tilde{\mu}\\
\geq& \kappa_1\kappa_2\kappa_3^2
\end{align*}
as the convexity of $C$ ensures that $\mathbf{1}_{\{\{s_\alpha(\mu)+s_\alpha(\tilde{\mu})\}/2 \in C\} } \geq \mathbf{1}_{s_\alpha(\mu) \in C}\mathbf{1}_{s_\alpha(\tilde{\mu}) \in C}$.
\end{proof}

We can remove the assumption that $\mathcal{B}$ is compact, by imposing a different assumption:

\begin{theorem}
\label{TH4}
Assume that there exist $\mathcal{H} \subset \mathcal{P}(\mathcal{A})$ stable by $Q$ and $A \subset \mathcal{A}$ and $C \subset s_\alpha(\mathcal{B})$ with finite positive measure such that:
\begin{align*}
\kappa_1&=\inf_{s_\alpha(\mu) \in C} \pi(B_{s_\alpha(\mu),\varepsilon_\alpha/4})>0; \\
\kappa_2&=\inf_{\alpha \in A} \inf_{s_\alpha(\mu) \in C} \pi(B_{s_\alpha(\mu),3\varepsilon_\alpha/2} \mid \alpha,x)>0;\\
\kappa_3 &= \inf_{\alpha \in A} \pi\{s_\alpha(\mu) \in C \mid \alpha,x \} > 0;\\
\kappa_4&=\inf_{\nu \in \mathcal{H}} \nu(A)>0.
\end{align*}
Then, the Markov chain associated with Algorithm \ref{algohierar} enjoys an
invariant distribution and it converges geometrically to this invariant measure
with rate $1-\kappa_1\kappa_2\kappa_3^2\kappa_4^2 $.
\end{theorem}

\begin{proof}

Similarly to previous proofs, we have
\begin{align*}
I_3 &\geq \int \mathrm{pr}(\mu \in B_{\frac{s_\alpha(\mu)+s_\alpha(\tilde{\mu})}{2},\varepsilon_\alpha/2})\pi(\mu \mid \alpha,x)\pi(\tilde{\mu} \mid \tilde{\alpha},x)\nu(\alpha)\tilde{\nu}(\tilde{\alpha})\\
&\hspace{1cm}\times\mathbf{1}_{s_\alpha(\mu)\in C}\mathbf{1}_{s_\alpha(\tilde{\mu}) \in C}\mathbf{1}_{\{\eta\{s_\alpha(\mu),s_\alpha(\tilde{\mu})\}>3\varepsilon_\alpha/2\}}\mathrm{d}\alpha \mathrm{d} \tilde{\alpha}\mathrm{d}\mu \mathrm{d}\tilde{\mu}\\
&\geq \int \kappa_1\kappa_2\kappa_3^2 \nu(\alpha)\tilde{\nu}(\tilde{\alpha})\mathbf{1}_{\alpha \in A} \mathbf{1}_{\tilde{\alpha} \in A} \\
&\geq \kappa_1\kappa_2\kappa_3^2\kappa_4^2\,. \\
\end{align*}
\end{proof}

\section{Supplementary material: Counter-example to Theorem \ref{THgen}}\label{sec:counter}

In this section, we give a simple example where the assumptions of Theorem \ref{THgen} are not verified and where \ABCG fails (whereas Vanilla ABC does not).

Take a single observation from a mixture of two uniforms, with parameterized by $(\theta_1, \theta_2)$:
\[x\sim \frac12 \U(\theta_1,\theta_1+1)+\frac12 \U(\theta_2,\theta_2+1).\]
For the numerical applications, we shall use the realization $x^\star=5$. Consider the prior distribution 
 \[(\theta_1,\theta_2)\sim\U(\mathcal A) 
 \qquad
  \mathcal A=\left\{ (\theta_1,\theta_2):0\leq\theta_1,\theta_2\leq10 \text{ and } |\theta_1-\theta_2| >2    \right\} \]
  
  The exact posterior is uniform over the set 
  \[ \left\{ \left([0,10]\times[x-1,x]\right)\cup \left([x-1,x]\times[0,10]\right) \right\}\cap\mathcal A. \]
  
  The prior and exact posterior are shown in Figure \ref{fig:counter}, as well as the outcome of Vanilla ABC and \ABCG with $\epsilon=\epsilon_1=\epsilon_2=0.5$.  Vanilla ABC leads to a reasonable approximation of the posterior, but \ABCG misses half of the posterior. Other realizations of \ABCG lead to the symmetric pseudo-posterior, with the roles of $\theta_1$ and $\theta_2$ swapped. This is a situation where the \ABCG does not converge to a unique stationary distribution $\nu_\epsilon$ (as soon as $\epsilon_1,\epsilon_2\leq\frac12$).
  
For Theorem \ref{THgen} to apply, we would need
\[
\sup_{\theta_{1},\tilde\theta_{1}} \Vert \pi_{\varepsilon_{2}}\{ \cdot \mid s_{2}(x^\star,  \theta_{1})\} - \pi_{\varepsilon_{2}}\{ \cdot \mid s_{2}(x^\star,  \tilde\theta_{1})\} \Vert_{TV} = \kappa < \frac12.
\]
  
  Consider $\theta_1=1$ and $\tilde\theta_1=5$. Then $\pi_{\varepsilon_{2}}\{ \cdot \mid s_{2}(x^\star,  \theta_{1})\}$ has support $[3.5, 5.5]$ and $\pi_{\varepsilon_{2}}\{ \cdot \mid s_{2}(x^\star,  \tilde\theta_{1})\}$ has support $[0,3]\cup[7,10]$. Since the two supports are disjoint, the distance in total variation between the two distributions is 1, and Theorem \ref{THgen} does not apply. Intuitively, the Markov chain does not converge because it is not irreducible.
  
  \begin{figure}
\includegraphics[scale=0.4]{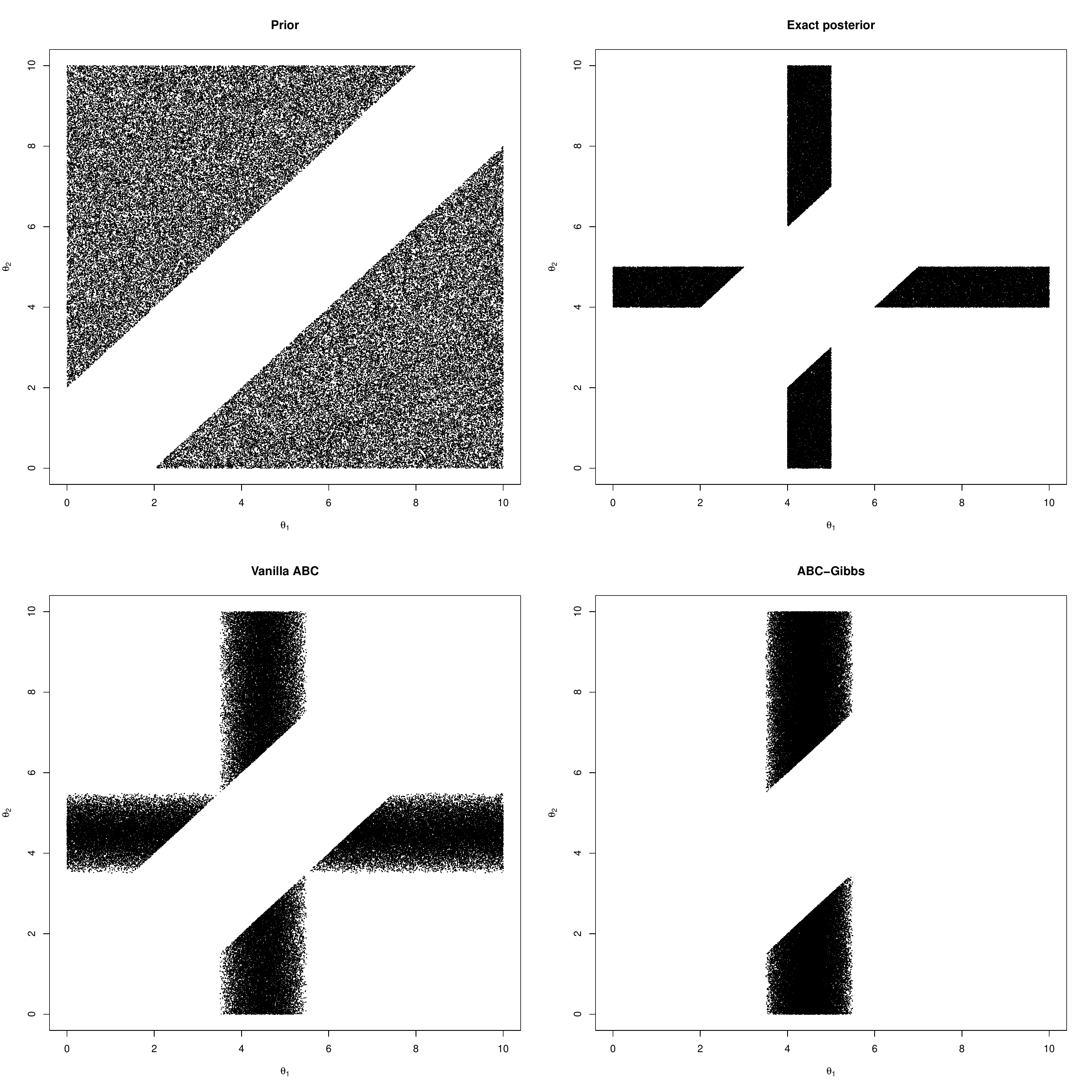}

\caption{\label{fig:counter} Illustration of the mixture of uniforms counter-example from Section \ref{sec:counter}, with $x^\star=5$ and $\epsilon=0.5$. Top left: prior distribution. Top right: Exact posterior. Bottom left: Vanilla ABC posterior. Bottom right: one possible outcome of \ABCG. The Vanilla ABC is a reasonable approximation of the exact posterior, but the \ABCG outcome only covers half of the support.}
\end{figure}


\begin{thebibliography}{31}
\expandafter\ifx\csname natexlab\endcsname\relax\def\natexlab#1{#1}\fi

\bibitem[{Arnold \& Press(1989)}]{arnold1989compatible}
\textsc{Arnold, B.~C.} \& \textsc{Press, S.~J.} (1989).
\newblock Compatible conditional distributions.
\newblock \textit{Journal of the American Statistical Association} \textbf{84},
  152--156.

\bibitem[{Beaumont et~al.(2009)Beaumont, Cornuet, Marin \&
  Robert}]{beaumont2009}
\textsc{Beaumont, M.~A.}, \textsc{Cornuet, J.-M.}, \textsc{Marin, J.-M.} \&
  \textsc{Robert, C.~P.} (2009).
\newblock Adaptive approximate {B}ayesian computation.
\newblock \textit{Biometrika} \textbf{96}, 983–990.

\bibitem[{Beaumont et~al.(2002)Beaumont, Zhang \& Balding}]{beaumont2002}
\textsc{Beaumont, M.~A.}, \textsc{Zhang, W.} \& \textsc{Balding, D.~J.} (2002).
\newblock Approximate {B}ayesian {C}omputation in {P}opulation {G}enetics.
\newblock \textit{Genetics} \textbf{162}, 2025--2035.

\bibitem[{Carlin \& Louis(1996)}]{carlin:louis:1996}
\textsc{Carlin, B.} \& \textsc{Louis, T.} (1996).
\newblock \textit{{B}ayes and Empirical {B}ayes Methods for Data Analysis}.
\newblock London: Chapman and Hall.

\bibitem[{Csill{\'e}ry et~al.(2010)Csill{\'e}ry, Blum, Gaggiotti \&
  Fran{\c{c}}ois}]{csillery2010}
\textsc{Csill{\'e}ry, K.}, \textsc{Blum, M. G.~B.}, \textsc{Gaggiotti, O.~E.}
  \& \textsc{Fran{\c{c}}ois, O.} (2010).
\newblock {Approximate {B}ayesian computation (ABC) in practice}.
\newblock \textit{Trends in Ecology \& Evolution} \textbf{25}, 410--418.

\bibitem[{Del~Moral et~al.(2012)Del~Moral, Doucet \&
  Jasra}]{delmoral:doucet:jasra:2012}
\textsc{Del~Moral, P.}, \textsc{Doucet, A.} \& \textsc{Jasra, A.} (2012).
\newblock An adaptive sequential {M}onte {C}arlo method for approximate
  {B}ayesian computation.
\newblock \textit{Statistics and Computing} \textbf{22}, 1009--1020.

\bibitem[{Fearnhead \& Prangle(2012)}]{fearnhead2012}
\textsc{Fearnhead, P.} \& \textsc{Prangle, D.} (2012).
\newblock Constructing summary statistics for approximate {B}ayesian
  computation: semi-automatic approximate {B}ayesian computation.
\newblock \textit{Journal of the Royal Statistical Society. Series B}
  \textbf{74}, 419--474.

\bibitem[{Frazier et~al.(2018)Frazier, Martin, Robert \&
  Rousseau}]{frazier2018}
\textsc{Frazier, D.}, \textsc{Martin, G.}, \textsc{Robert, C.} \&
  \textsc{Rousseau, J.} (2018).
\newblock Asymptotic properties of approximate {B}ayesian computation.
\newblock \textit{Biometrika} \textbf{105}, 593–607.

\bibitem[{Gelfand \& Smith(1990)}]{GelflandSmith}
\textsc{Gelfand, A.~E.} \& \textsc{Smith, A. F.~M.} (1990).
\newblock {Sampling-Based Approaches to Calculating Marginal Densities}.
\newblock \textit{Journal of the American Statistical Association} \textbf{85},
  398--409.

\bibitem[{Gelman et~al.(2013)Gelman, Carlin, Stern, Dunson, Vehtari \&
  Rubin}]{gelman2013bayesian}
\textsc{Gelman, A.}, \textsc{Carlin, J.~B.}, \textsc{Stern, H.~S.},
  \textsc{Dunson, D.~B.}, \textsc{Vehtari, A.} \& \textsc{Rubin, D.~B.} (2013).
\newblock \textit{{Bayesian Data Analysis, Third Edition}}.
\newblock Chapman \& Hall/CRC Texts in Statistical Science. Taylor \& Francis.

\bibitem[{Geman \& Geman(1984)}]{geman:1984}
\textsc{Geman, S.} \& \textsc{Geman, D.} (1984).
\newblock Stochastic relaxation, {G}ibbs distributions and the {B}ayesian
  restoration of images.
\newblock \textit{IEEE Trans. Pattern Anal. Mach. Intell.} \textbf{6},
  721--741.

\bibitem[{Kaipio \& Fox(2011)}]{Kaipio:2011}
\textsc{Kaipio, J.~P.} \& \textsc{Fox, C.} (2011).
\newblock {The {B}ayesian Framework for Inverse Problems in Heat Transfer}.
\newblock \textit{Heat Transfer Engineering} \textbf{32}, 718--753.

\bibitem[{Kousathanas et~al.(2016)Kousathanas, Leuenberger, Helfer, Quinodoz,
  Foll \& Wegmann}]{kousathanas:2016}
\textsc{Kousathanas, A.}, \textsc{Leuenberger, C.}, \textsc{Helfer, J.},
  \textsc{Quinodoz, M.}, \textsc{Foll, M.} \& \textsc{Wegmann, D.} (2016).
\newblock Likelihood-free inference in high-dimensional models.
\newblock \textit{Genetics} \textbf{203}, 893--904.

\bibitem[{Li \& Fearnhead(2018{\natexlab{a}})}]{li2018convergence}
\textsc{Li, W.} \& \textsc{Fearnhead, P.} (2018{\natexlab{a}}).
\newblock Convergence of regression-adjusted approximate {B}ayesian
  computation.
\newblock \textit{Biometrika} \textbf{105}, 301--318.

\bibitem[{Li \& Fearnhead(2018{\natexlab{b}})}]{li2018}
\textsc{Li, W.} \& \textsc{Fearnhead, P.} (2018{\natexlab{b}}).
\newblock On the asymptotic efficiency of approximate {B}ayesian computation
  estimators.
\newblock \textit{Biometrika} \textbf{105}, 285--299.

\bibitem[{Lindley \& Smith(1972)}]{lindley:smith:1972}
\textsc{Lindley, D.} \& \textsc{Smith, A.} (1972).
\newblock {B}ayes estimates for the linear model.
\newblock \textit{Journal of the Royal Statistical Society. Series B}
  \textbf{34}, 1--41.

\bibitem[{Lunn et~al.(2010)Lunn, Thomas, Best \&
  Spiegelhalter}]{lunn:bugs:2012}
\textsc{Lunn, D.}, \textsc{Thomas, A.}, \textsc{Best, N.} \&
  \textsc{Spiegelhalter, D.} (2010).
\newblock \textit{The {BUGS} Book: A Practical Introduction to {B}ayesian
  Analysis}.
\newblock New York: Chapman \& Hall/CRC Press.

\bibitem[{Marin et~al.(2012)Marin, Pudlo, Robert \& Ryder}]{Marin2012}
\textsc{Marin, J.-M.}, \textsc{Pudlo, P.}, \textsc{Robert, C.~P.} \&
  \textsc{Ryder, R.~J.} (2012).
\newblock {Approximate {B}ayesian computational methods}.
\newblock \textit{Statistics and Computing} \textbf{22}, 1167--1180.

\bibitem[{Marjoram et~al.(2003)Marjoram, Molitor, Plagnol \&
  Tavar{\'e}}]{marjoram2003}
\textsc{Marjoram, P.}, \textsc{Molitor, J.}, \textsc{Plagnol, V.} \&
  \textsc{Tavar{\'e}, S.} (2003).
\newblock Markov chain {M}onte {C}arlo without likelihoods.
\newblock \textit{Proceedings of the National Academy of Sciences}
  \textbf{100}, 15324--15328.

\bibitem[{Moores et~al.(2015)Moores, Drovandi, Mengersen \&
  Robert}]{moores2015}
\textsc{Moores, M.~T.}, \textsc{Drovandi, C.~C.}, \textsc{Mengersen, K.} \&
  \textsc{Robert, C.~P.} (2015).
\newblock Pre-processing for approximate {B}ayesian computation in image
  analysis.
\newblock \textit{Statistics and Computing} \textbf{25}, 23--33.

\bibitem[{Neal(2012)}]{neal2012efficient}
\textsc{Neal, P.} (2012).
\newblock {Efficient likelihood-free {B}ayesian Computation for household
  epidemics}.
\newblock \textit{Statistics and Computing} \textbf{22}, 1239--1256.

\bibitem[{Nott et~al.(2014)Nott, Fan, Marshall \& Sisson}]{nott2014approximate}
\textsc{Nott, D.~J.}, \textsc{Fan, Y.}, \textsc{Marshall, L.} \&
  \textsc{Sisson, S.} (2014).
\newblock Approximate {B}ayesian computation and {B}ayes’ linear analysis:
  toward high-dimensional {ABC}.
\newblock \textit{Journal of Computational and Graphical Statistics}
  \textbf{23}, 65--86.

\bibitem[{Prangle(2017)}]{prangle2017gk}
\textsc{Prangle, D.} (2017).
\newblock gk: An {R} package for the $g$-and-$k$ and generalised $g$-and-$h$
  distributions.
\newblock \textit{arXiv preprint arXiv:1706.06889} .

\bibitem[{Raynal et~al.(2019)Raynal, Marin, Pudlo, Ribatet, Robert \&
  Estoup}]{ABCRF}
\textsc{Raynal, L.}, \textsc{Marin, J.-M.}, \textsc{Pudlo, P.},
  \textsc{Ribatet, M.}, \textsc{Robert, C.~P.} \& \textsc{Estoup, A.} (2019).
\newblock {ABC random forests for {B}ayesian parameter inference}.
\newblock \textit{Bioinformatics} \textbf{35}, 1720–1728.

\bibitem[{Robert \& Casella(2004)}]{robert:casella:2004}
\textsc{Robert, C.} \& \textsc{Casella, G.} (2004).
\newblock \textit{{M}onte {C}arlo Statistical Methods}.
\newblock New York: Springer Verlag, 2nd ed.

\bibitem[{Rodrigues et~al.(2020)Rodrigues, Nott \&
  Sisson}]{rodrigues2019likelihood}
\textsc{Rodrigues, G.}, \textsc{Nott, D.~J.} \& \textsc{Sisson, S.} (2020).
\newblock Likelihood-free approximate {G}ibbs sampling.
\newblock \textit{Statistics and Computing} \textbf{30}, 1057--1073.

\bibitem[{Sisson et~al.(2018)Sisson, Fan \& Beaumont}]{handbook}
\textsc{Sisson, S.}, \textsc{Fan, Y.} \& \textsc{Beaumont, M.}, eds. (2018).
\newblock \textit{Handbook of {A}pproximate {B}ayesian {C}omputation}.
\newblock New York: Chapman and Hall/CRC.

\bibitem[{Tavar{\'e} et~al.(1997)Tavar{\'e}, Balding, Griffiths \&
  Donnelly}]{tavare1997}
\textsc{Tavar{\'e}, S.}, \textsc{Balding, D.~J.}, \textsc{Griffiths, R.~C.} \&
  \textsc{Donnelly, P.} (1997).
\newblock Inferring {C}oalescence {T}imes {F}rom {DNA} {S}equence {D}ata.
\newblock \textit{Genetics} \textbf{145}, 505--518.

\bibitem[{Toni et~al.(2008)Toni, Welch, Strelkowa, Ipsen \& Stumpf}]{toni2008}
\textsc{Toni, T.}, \textsc{Welch, D.}, \textsc{Strelkowa, N.}, \textsc{Ipsen,
  A.} \& \textsc{Stumpf, M. P.~H.} (2008).
\newblock {Approximate {B}ayesian computation scheme for parameter inference
  and model selection in dynamical systems}.
\newblock \textit{Journal of the Royal Society Interface} \textbf{6}, 187--202.

\bibitem[{Turner \& Van~Zandt(2013)}]{turner2013hierarchical}
\textsc{Turner, B.} \& \textsc{Van~Zandt, T.} (2013).
\newblock Hierarchical approximate {B}ayesian computation.
\newblock \textit{Psychometrika} \textbf{79}.

\bibitem[{Wilkinson et~al.(2011)Wilkinson, Steiper, Soligo, Martin, Yang \&
  Tavar{\'e}}]{wilkinson:2011}
\textsc{Wilkinson, R.}, \textsc{Steiper, M.}, \textsc{Soligo, C.},
  \textsc{Martin, R.}, \textsc{Yang, Z.} \& \textsc{Tavar{\'e}, S.} (2011).
\newblock Dating primate divergences through an integrated analysis of
  palaeontological and molecular data.
\newblock \textit{Systematic Biology} \textbf{60}, 16–31.

\end{thebibliography}

\begin{thebibliography}{5}
\expandafter\ifx\csname natexlab\endcsname\relax\def\natexlab#1{#1}\fi

\bibitem[{Del~Moral et~al.(2012)Del~Moral, Doucet \&
  Jasra}]{delmoral:doucet:jasra:2012}
\textsc{Del~Moral, P.}, \textsc{Doucet, A.} \& \textsc{Jasra, A.} (2012).
\newblock An adaptive sequential {M}onte {C}arlo method for approximate
  {B}ayesian computation.
\newblock \textit{Statistics and Computing} \textbf{22}, 1009--1020.

\bibitem[{Lazio et~al.(2008)Lazio, B.~Waltman, D.~Ghigo, Fiedler, S.~Foster \&
  K.~J.~Johnston}]{Laziodualfrequency}
\textsc{Lazio, T. J.~W.}, \textsc{B.~Waltman, E.}, \textsc{D.~Ghigo, F.},
  \textsc{Fiedler, R.}, \textsc{S.~Foster, R.} \& \textsc{K.~J.~Johnston, a.}
  (2008).
\newblock {A Dual-Frequency, Multiyear Monitoring Program of Compact Radio
  Sources}.
\newblock \textit{The Astrophysical Journal Supplement Series} \textbf{136},
  265.

\bibitem[{Meyn \& Tweedie(1993)}]{meyn:twee:1993}
\textsc{Meyn, S.~P.} \& \textsc{Tweedie, R.~L.} (1993).
\newblock \textit{Markov Chains and Stochastic Stability}.
\newblock London: Springer-Verlag.

\bibitem[{Nummelin(1978)}]{nummelin:1978}
\textsc{Nummelin, E.} (1978).
\newblock A splitting technique for {H}arris recurrent chains.
\newblock \textit{Zeit. Warsch. Verv. Gebiete} \textbf{43}, 309--318.

\bibitem[{Toni et~al.(2008)Toni, Welch, Strelkowa, Ipsen \& Stumpf}]{toni2008}
\textsc{Toni, T.}, \textsc{Welch, D.}, \textsc{Strelkowa, N.}, \textsc{Ipsen,
  A.} \& \textsc{Stumpf, M. P.~H.} (2008).
\newblock {Approximate {B}ayesian computation scheme for parameter inference
  and model selection in dynamical systems}.
\newblock \textit{Journal of the Royal Society Interface} \textbf{6}, 187--202.

\end{thebibliography}

\end{document}